\documentclass{article}%
\usepackage{amsmath}
\usepackage{amsmath,extarrows}
\usepackage{amsfonts}
\usepackage{amssymb}
\usepackage{amsthm}
\usepackage{graphicx}
\usepackage{epsfig}
\usepackage{mathdots}
\usepackage{xcolor}
\usepackage{hyperref}
\usepackage{rotating}
\usepackage{enumerate}
\usepackage{tikz}
\usepackage{multirow}
\usepackage{fullpage}%
\providecommand{\U}[1]{\protect\rule{.1in}{.1in}}
\allowdisplaybreaks
\usetikzlibrary{positioning, shapes.geometric}
\tikzset{global scale/.style={scale=#1,every node/.append style={scale=#1}}}

\usetikzlibrary{arrows}
\newtheorem{theorem}{Theorem}

\newtheorem{algorithm}[theorem]{Algorithm}

\newtheorem{definition}[theorem]{Definition}
\newtheorem{example}[theorem]{Example}

\newtheorem{lemma}[theorem]{Lemma}
\newtheorem{notation}[theorem]{Notation}
\newtheorem{problem}[theorem]{Problem}
\newtheorem{proposition}[theorem]{Proposition}
\newtheorem{remark}[theorem]{Remark}

\begin{document}

\title{Computing the greatest common divisor \\of\ several parametric univariate polynomials \\via generalized subresultants}
\author{Hoon Hong\\Department of Mathematics, North Carolina State University\\Box 8205, Raleigh, NC 27695, USA\\hong@ncsu.edu\\[10pt] Jing Yang\thanks{Corresponding author.}\\SMS--HCIC--School of Mathematics and Physics,\\Center for Applied Mathematics of Guangxi,\\Guangxi Minzu University, Nanning 530006, China\\yangjing0930@gmail.com\\[-10pt]}
\date{}
\maketitle

\begin{abstract}
In this paper, we tackle the following problem: compute the gcd for
\emph{several} univariate polynomials with \emph{parametric} coefficients. It
amounts to partitioning the parameter space into ``cells'' so that the gcd has
a uniform expression over each cell and to constructing a uniform expression
of gcd in each cell. We tackle the problem as follows. We begin by making a
natural and obvious extension of subresultants of two polynomials to several
polynomials. Then we develop the following structural theories about them.

\begin{enumerate}
\item We generalize Sylvester's theory to several polynomials, in order to
obtain an elegant relationship between generalized subresultants and the gcd
of several polynomials, yielding an elegant algorithm.

\item We generalize Habicht's theory to several polynomials, in order to
obtain a systematic relationship between generalized subresultants and
pseudo-remainders, yielding an efficient algorithm.

\end{enumerate}

\noindent Using the generalized theories, we present a simple (structurally
elegant) algorithm which is significantly more efficient (both in the output
size and computing time) than algorithms based on previous approaches.

\end{abstract}

\section{Introduction}

\label{sec:introduction} The greatest common divisor (gcd) for two univariate
polynomials is a fundamental object in computational algebra and geometry with
numerous applications in science and engineering. Due to its importance, there
have been extensive research on underlying theories and algorithms for
computing the gcd of given two polynomials. In this paper, we tackle the
following generalization of the gcd problem: compute the gcd for
\emph{several} univariate polynomials with \emph{parametric} coefficients.

The motivation for considering such a generalization of gcd comes from two
sources: (1) One often needs to work with more than two polynomials. For
instance, when studying the multiplicity structure of the roots of a
polynomial, it is natural to consider its derivatives of several orders. (2)
One often wants to carry out structural analysis for a parameterized family of inputs.

\medskip Obviously, the gcd depends on values of the parameters. Hence the
problem essentially consists of the following two challenges:

\begin{enumerate}
\item[C1:\ \;] Partition the parameter space into ``cells'' so that the gcd
has a uniform expression over each cell.

To be more specific, find a system of polynomial equations and inequalities in
the parameters that describe each cell.

\item[C2:\ \;] Construct a uniform expression of gcd in each cell.

To be more specific, find a univariate polynomial for the gcd where its
coefficients are rational functions in the parameters.
\end{enumerate}

\noindent For the case of two polynomials, there have been extensive research
on the above two challenges (i.e., C1 and C2) during the last two centuries,
producing deep theories and efficient algorithms as well as numerous
applications (just list a few, \cite{1853_Sylvester, 
2003_Lascoux_Pragacz,
1907_Bocher, 
1968_Householder, 
1967_Collins, 
1971_BARNETT, 
1971_Brown_Traub,
1997_Hong, 
1998_Wang, 
Hong:99a, 
2000_Wang, 
2004_WANG_HOU, 
2008_Akira,
2004_DIAZ_TOCA_GONZALEZ_VEGA, 
2007_DAndrea_Hong_Krick_Szanto,
2006_DAndrea_Krick_Szanto, 
2009_DAndrea_Hong_Krick_Szanto,
2015_DAndrea_Krick_Szanto, 
2017_Krick_Szanto_Valdettaro,
2017_Bostan_DAndrea_Krick_Szanto_Valdettaro,
2019_Andrea_Krick_Szanto_Valdettaro, 
2020_Bostan_Krick_Szanto_Valdettaro,
Wang:2024,
Li:98, 
Hong:99b, 
Hong:2000a, 
2001_DAndrea_Dickenstein,
2013_DAndrea_Krick_Szanto, 
2020_Roy_Szpirglas, 
2023_Cox_DAndrea,
2013_Jaroschek} and \cite{ 
2008_Szanto, 
2017_Imbach_Moroz_Pouget,
2018_Perrucci_Roy, 
2020_Roy_Szpirglas}). By now, the most successful approach
seems to be the integration of the following two theories:

\begin{itemize}
\item Sylvester's theory on an elegant relationship between subresultants and
the gcd of two given polynomials, yielding an elegant algorithm;

\item Habicht's theory on a systematic relationship between subresultants and
pseudo-remainders, yielding an efficient algorithm.
\end{itemize}

\noindent They together tackle the above two challenges (C1 and C2) elegantly
and efficiently.

\bigskip{\ For the case of more than two polynomials, there have been sporadic
research on the above two challenges (C1 and C2). They can be roughly
classified into two kinds of strategies: (1) recursive approach and (2)
generalized-resultant-matrix-based approach. Below we briefly describe, for
each approach, the main underlying ideas, the main advantages, and some
challenges for further improvements. }

\begin{itemize}
\item \emph{Recursive approach} \cite{1996_Yang_Hou_Zeng}. This approach
recursively employs the subresultant algorithm for two polynomials (see
Algorithm \ref{alg:recursive} in the present paper). The main advantage of
this approach is that the number of polynomials describing the cells is small.
On the other hand, due to the recursive process, this approach often leads to
nested determinants, \textit{i.e.}, determinants of determinants of ..., and
so on, which hinder further structural analysis. Furthermore, the polynomials
describing the cells could have unnecessarily high degrees and the process
could be very time-consuming. Hence the challenge is to produce ``flat''
(non-nested) determinants with smaller degrees.

\item \emph{Generalized-resultant-matrix-based approach}
\cite{1971_BARNETT,1976_Kakie,1978_Vardulakis_Stoyle,Ho:89,2002_Diaz_Toca_Gonzalez_Vega}%
. The subresultants of two polynomials can be expressed as the determinants of
various ``resultant'' matrices (\cite{1853_Sylvester,1857_Cayley,1970_Barnett}%
). Naturally there have been efforts to generalize the resultant matrices to
several polynomials. In 1971, Barnett first generalized the companion
resultant matrix of two polynomials to several polynomials \cite{1971_BARNETT}%
. Later, Kaki\'{e}, Vardulakis and Stoyle proposed two different approaches
for the generalization of Sylvester matrix
\cite{1976_Kakie,1978_Vardulakis_Stoyle}. In 1989, Ho made a further
investigation on Kaki\'{e}'s result and presented a parallel algorithm for
computing the gcd of several univariate polynomials with constant
coefficients~\cite{Ho:89}. In 2002, Diaz-Toca and Gonzalez-Vega, inspired by
Barnett's idea in~\cite{1971_BARNETT}, extended the construction of B\'ezout
matrix, hybrid B\'ezout matrix and Hankel matrix to several polynomials in
\cite{2002_Diaz_Toca_Gonzalez_Vega}. With these generalized resultant
matrices, one can partition the parameter space into cells by rank conditions
and provide a uniform expression for the gcd over each cell. The main
advantage of this approach is that it leads to ``flat'' (non-nested)
determinants, which aid further structural analysis. On the other hand, one
observes that the number of polynomials describing the cells could be very
large. The challenge is to produce cell descriptions that involve smaller
number of polynomials.

\end{itemize}

\noindent Thus there is a need for developing a new approach that tackles the
challenges of the previous approaches while keeping their advantages.

\bigskip In this paper, we develop such an approach in that it produces flat
(non-nested) determinants and that the number and the degrees of polynomials
describing the cells are small. We begin by making a natural and
straightforward extension of subresultants of two polynomials to several
polynomials (Definition~\ref{def:sres_npoly} in Section~\ref{sec:gen_sres}).
Then we develop the following structural theories about them.

\begin{enumerate}
\item We \emph{generalize Sylvester's theory to several polynomials}
(Section~\ref{sec:gen_sres}), in order to obtain an elegant relationship
between generalized subresultants and the gcd of several polynomials, yielding
an \emph{elegant algorithm}.


We extend the parametric conditions and expressions for gcd of two polynomials
to multiple polynomials and give the parametric conditions and gcd expressions
in the form of formulas. Once the polynomials are given, one may immediately
write down the conditions and gcd expressions without any further mental work.
Moreover, the conditions and gcd expressions are all non-nested determinants
or determinant polynomials, which makes them easier for further analysis.

See Theorem~\ref{thm:gsylvester} and Algorithm \ref{alg:non_recursive_naive}
for details.

\item We \emph{generalize Habicht's theory to several polynomials}
(Section~\ref{sec:gen_habicht}), in order to obtain a systematic relationship
between generalized subresultants and pseudo-remainders, yielding an
\emph{efficient algorithm}.


The most time-consuming part of the algorithm obtained from the generalized
Sylvester's theory is the computation of determinant polynomials. For this
purpose, we
reveal an inherent relation among the subresultants appearing in the
algorithm. Although the relation is not a direct generalization of the
classical Habicht's theorem, they share the same idea, that is, converting the
computation of subresultants into the pseudo-remainder calculation. By
combining the classical Habicht's theorem and the newly revealed relation
together, we design an efficient algorithm for computing the parametric gcds
for a set of parametric polynomials.

See Theorem~\ref{thm:ghabicht} and Algorithm \ref{alg:non_recursive_cute} for details.
\end{enumerate}

\noindent Finally we compare the resulting algorithm with the classic one,
which shows that the proposed algorithm is significantly more efficient, both
in the output size and computing time (Section~\ref{sec:performance}).

\bigskip\noindent\textit{Comparison with related works}: There have been
numerous works on computing gcds for polynomials. We will focus on efforts in
parametric gcd for several univariate polynomials.

\begin{itemize}
\item In \cite{1976_Kakie,1978_Vardulakis_Stoyle}, the authors generalized the
following well-known property of the subresultant of two univariate
polynomials: the degree of the gcd of two polynomials is determined by the
rank of the corresponding Sylvester matrix . They extended the concept of
Sylvester matrix to several univariate polynomials so that a similar property holds.

In this paper, we generalize another well-known property of the subresultant
of two univariate polynomials: the degree of the gcd of two polynomials is
determined by vanishing of principal coefficient of subresultants and the gcd
is given by the corresponding subresultant. We extend subresultant to several
polynomials so that a similar property holds.

\item In \cite{Ho:89}, the author further developed the work in
\cite{1976_Kakie} and found one way to identify a maximal set of linearly
independent rows in the subresultant matrices from \cite{1976_Kakie}. This
result was used to investigate the problem of computing the gcd for several
univariate polynomials with numeric coefficients.

In this paper, we further refine the result in the sense that we allow the
coefficients to be parameters and provide an explicit formula to write down
the gcd in terms of coefficients of the given polynomials.

\item Habicht's theorem, proposed by W. Habicht in \cite{Habicht:48}, is one
of the essential ingredients in the classical subresultant theory. It captures
the fundamental relationships among three consecutive subresultants in the
subresultant chain of two parametric polynomials with formal coefficients and
lays the foundation for validating the gap structure of subresultant chain. It
also plays a vital role in the development of subresultant theory and its
subsequent applications. With the help of Habicht's theorem, one can compute
the parametric gcds of two parametric univariate polynomials very efficiently.

In this paper, we propose an analog of Habicht's theorem for subresultants of
several univariate polynomials with formal coefficients, which can be used to
enhance the efficiency of the algorithm for computing parametric gcds of
several parametric univariate polynomials.
\end{itemize}

The paper is structured as follows. In Section~\ref{sec:problem}, we state the
problem precisely. In Section \ref{sec:review}, we briefly review the
subresultant theory of two univariate polynomials due to Sylvester and
Habicht. In Section~\ref{sec:gen_sres}, we describe a natural and obvious
extension of the classical subresultant for two polynomials to several
polynomials. In Section~\ref{sec:gen_sres}, we generalize Sylvester's theory
to several polynomials and end with an elegant algorithm. In
Section~\ref{sec:gen_habicht}, we generalize Habicht's theory to several
polynomials, and end with an efficient algorithm. In
Section~\ref{sec:performance}, we compare the performance of the proposed
algorithm and the previous approaches.

\section{Problem statement}

In this section, we give a precise statement of the problem. For this, we need
some notations.

\label{sec:problem}

\begin{notation}
[Global Notation]\label{notation1} We will use the following notation
throughout the paper.

\begin{itemize}
\item $d=\left(  d_{0},\ldots,d_{n}\right)  \in\mathbb{N}_{>0}^{n+1}$ be such
that $d_{0}\leq d_{1},\ldots,d_{n}$


\item $a_{i}=\left(  a_{i0},\ldots,a_{id_{i}}\right)  $ be indeterminates (parameters);

\item $F=(F_{0},\ldots,F_{n})$ where $F_{i}=\sum_{j=0}^{d_{i}}a_{ij}x^{j}%
\in\mathbb{Z}[a_{i}][x]$;
\end{itemize}
\end{notation}

\begin{remark}
\ 

\begin{enumerate}
\item The condition that $d_{0}\leq d_{1},\ldots,d_{n}$ is a natural
assumption for the parametric gcd problem considered in this paper because the
degree of gcd of several polynomials does not surpass the minimal degree of
the input polynomials.

\item The coefficients $a_{ij}$'s are indeterminates. When they take different
values, the gcd of $F$ may have different forms.
\end{enumerate}
\end{remark}

\noindent Now we are ready to give a precise statement of the problem.

\begin{problem}
[Parametric gcd problem]The parametric gcd problem is stated as:

\begin{enumerate}
\item[In\ :] $d$, a vector of positive integers, standing for the degrees of
$F$ (See Notation~\ref{notation1}).

\item[Out:] $\mathcal{G}$ is a representation of the parametric gcd for $F$,
that is,%
\[
\mathcal{G}=\left(  \left(  C_{1},G_{1}\right)  ,\ldots,\left(  C_{p}%
,G_{p}\right)  \right)
\]
such that%
\[
\gcd\left(  F\right)  =\left\{
\begin{array}
[c]{lll}%
G_{1} & \text{if} & C_{1}\\
G_{2} & \text{if} & C_{2}\\
\ \vdots & \ \vdots & \ \vdots\\
G_{p} & \text{if} & C_{p}%
\end{array}
\right.
\]
where $C_{i}$ is a condition on the parameters $a$, describing a cell, and
where $G_{i}$ is a parametric expression for the gcd of~$F$ over the cell.
\end{enumerate}
\end{problem}

\begin{example}
\ 

\begin{enumerate}
\item[In\ :] $d = (3,3,4)$ representing $F=(F_{0},F_{1},F_{2})$ where
\begin{align*}
F_{0}  &  =a_{03}x^{3}+a_{02}x^{2}+a_{01}x+a_{00}\\
F_{1}  &  =a_{13}x^{3}+a_{12}x^{2}+a_{11}x+a_{10}\\
F_{2}  &  =a_{24}x^{4}+a_{23}x^{3}+a_{22}x^{2}+a_{21}x+a_{20}%
\end{align*}

\item[Out:] $\mathcal{G}=\left(  (G_{1},C_{1}),\ldots,(G_{10},C_{10})\right)
$ where%
\[%
\begin{array}
[c]{lll}%
G_{1}={\operatorname*{dp}M}_{1} &  & C_{1}:\det M_{1}\neq0\\
G_{2}={\operatorname*{dp}M}_{2} &  & C_{2}:\det M_{1}=0\ \wedge\ \det
M_{2}\neq0\\
G_{3}={\operatorname*{dp}M}_{3} &  & C_{3}:\det M_{1}=0\ \wedge\ \det
M_{2}=0\ \wedge\ \det M_{3}\neq0\\
\vdots &  & \vdots\\
G_{10}={\operatorname*{dp}M}_{10} &  & C_{10}:\det M_{1}=0\ \wedge
\ \cdots\ \wedge\ \det M_{9}=0\ \wedge\ \det M_{10}\neq0
\end{array}
\]

where again%
\begin{align*}
M_{1}  &  =%
\begin{bmatrix}
a_{03} & a_{02} & a_{01} & a_{00} &  & \\
& a_{03} & a_{02} & a_{01} & a_{00} & \\
&  & a_{03} & a_{02} & a_{01} & a_{00}\\
a_{13} & a_{12} & a_{11} & a_{10} &  & \\
& a_{13} & a_{12} & a_{11} & a_{10} & \\
&  & a_{13} & a_{12} & a_{11} & a_{10}%
\end{bmatrix}
\\
M_{2}  &  =%
\begin{bmatrix}
a_{03} & a_{02} & a_{01} & a_{00} & \\
& a_{03} & a_{02} & a_{01} & a_{00}\\
a_{13} & a_{12} & a_{11} & a_{10} & \\
& a_{13} & a_{12} & a_{11} & a_{10}\\
a_{24} & a_{23} & a_{22} & a_{21} & a_{20}%
\end{bmatrix}
\\
M_{3}  &  =%
\begin{bmatrix}
a_{03} & a_{02} & a_{01} & a_{00} &  & \\
& a_{03} & a_{02} & a_{01} & a_{00} & \\
&  & a_{03} & a_{02} & a_{01} & a_{00}\\
&  & a_{13} & a_{12} & a_{11} & a_{10}\\
a_{24} & a_{23} & a_{22} & a_{21} & a_{20} & \\
& a_{24} & a_{23} & a_{22} & a_{21} & a_{20}%
\end{bmatrix}
\\
&  \vdots\\
M_{10}  &  =%
\begin{bmatrix}
a_{03} & a_{02} & a_{01} & a_{00}%
\end{bmatrix}
\end{align*}
In the above ${\operatorname*{dp}}$ stands for the widely used notion of
determinant polynomial (see Definition \ref{def:dp} and the following example).
\end{enumerate}
\end{example}

\section{Review on subresultant theory of two polynomials}

\label{sec:review}

In this section, we review the notion of subresultants of two polynomials,
Sylvester's theorem and Habicht's theorem on them. We will do so using a
\textbf{new indexing} scheme for them. The reason is that the new indexing
will facilitate their generalization to several polynomials. Thus, we strongly
encourage the readers (even those readers who are familiar with the classical
theory) to read this section in order to get familiar with the new indexing
scheme. In the rest of the paper, let $\mathcal{Z}$ denote an integral domain
such as $\mathbb{Z}$, $\mathbb{Q}$, $\mathbb{Z}\left[  a\right]  $ and so on.


\subsection{Subresultants for two polynomials}

\begin{definition}
[Determinant polynomial of matrix]\label{def:dp}Let
$M\in\mathcal{Z}^{p\times q}$ where $p\leq q$ (that is, $M$
is square or wide).

\begin{itemize}
\item The \emph{determinant polynomial} of $M$, written as
${\operatorname*{dp}}(M)$, is defined by
\[
{\operatorname*{dp}}(M)=\sum_{0\leq j\leq q-p}c_{j}x^{j}%
\]
where $c_{j}=\det\left[  M_{1}\ \cdots\ M_{p-1}\ M_{q-j}\right]  $ and $M_{k}$
stands for the $k$-th column of $M$.

\item The \emph{principal coefficient} of $\operatorname*{dp}(M)$, written as
$\operatorname*{pcdp}(M)$, is defined by
\[
{\operatorname*{pcdp}}\left(  M\right)  =\ \operatorname*{coeff}%
\nolimits_{x^{q-p} }\left(  \operatorname*{dp}\left(  M\right)  \right)
\]

\end{itemize}
\end{definition}

\begin{example}
Let
\[
M=\left[
\begin{array}
[c]{ccccc}%
m_{11} & m_{12} & m_{13} & m_{14} & m_{15}\\
m_{21} & m_{22} & m_{23} & m_{24} & m_{25}\\
m_{31} & m_{32} & m_{33} & m_{34} & m_{35}%
\end{array}
\right]
\]
Note $p=3$ and $q=5$. Thus

\begin{itemize}
\item ${\operatorname*{dp}}\left(  M\right)  =c_{2}x^{2}+c_{1}x^{1}+c_{0}
x^{0}$ where
\begin{align*}
c_{2}  &  =\det\left[  M_{1}\ M_{2}\ M_{5-2}\right]  =\det\left[
\begin{array}
[c]{cc|c}%
m_{11} & m_{12} & m_{13}\\
m_{21} & m_{22} & m_{23}\\
m_{31} & m_{32} & m_{33}%
\end{array}
\right] \\
c_{1}  &  =\det\left[  M_{1}\ M_{2}\ M_{5-1}\right]  =\det\left[
\begin{array}
[c]{cc|c}%
m_{11} & m_{12} & m_{14}\\
m_{21} & m_{22} & m_{24}\\
m_{31} & m_{32} & m_{34}%
\end{array}
\right] \\
c_{0}  &  =\det\left[  M_{1}\ M_{2}\ M_{5-0}\right]  =\det\left[
\begin{array}
[c]{cc|c}%
m_{11} & m_{12} & m_{15}\\
m_{21} & m_{22} & m_{25}\\
m_{31} & m_{32} & m_{35}%
\end{array}
\right]
\end{align*}

\item ${\operatorname*{pcdp}}\left(  M\right)  =c_{2}=\det\left[  M_{1}
\ M_{2}\ M_{5-2}\right]  =\det\left[
\begin{array}
[c]{cc|c}%
m_{11} & m_{12} & m_{13}\\
m_{21} & m_{22} & m_{23}\\
m_{31} & m_{32} & m_{33}%
\end{array}
\right]  $
\end{itemize}
\end{example}

\begin{definition}
[Coefficient matrix of a list of polynomials]Let $P=(P_{1},\ldots,P_{t})$
where
\[
P_{i}=\sum_{0\le j\le p_{i}}b_{ij}x^{j}\in\mathcal{Z}[x]
\]
and $p_{i}=\deg P_{i}$. Let $m=\max_{1\le i\le t} p_{i}$. Then the
\emph{coefficient matrix} of $P$, written as
$\textcolor{red}{{\operatorname*{cm}}(P)}$, is defined as the $t\times(m+1)$
matrix whose $(i,j)$-th entry is the coefficient of $P_{i}$ in the term
$x^{m+1-j}$.
\end{definition}

\begin{example}
\label{example:P-cm(P)} Let $P=(P_{1},P_{2},P_{3})$ where
\begin{align*}
P_{1}  &  =b_{03}x^{3}+b_{02}x^{2}+b_{01}x+b_{00}\\
P_{2}  &  =b_{13}x^{3}+b_{12}x^{2}+b_{11}x+b_{10}\\
P_{3}  &  =b_{22}x^{2}+b_{21}x+b_{20}%
\end{align*}
Thus
\begin{align*}
\operatorname{cm}(P)=\operatorname{cm}(P_{1},P_{2},P_{3})= \left[
\begin{array}
[c]{cccc}%
b_{03} & b_{02} & b_{01} & b_{00}\\
b_{13} & b_{12} & b_{11} & b_{10}\\
& b_{22} & b_{21} & b_{20}%
\end{array}
\right]
\end{align*}

\end{example}

\begin{notation}
[Determinant polynomial of a list of polynomials]Let $P=(P_{1},\ldots,P_{t})$
be such that $\operatorname*{cm}(P)$ is square or wide. Then we will use the
following short hand notations.

\begin{itemize}
\item $\operatorname*{dp}(P)=\operatorname*{dp}(\operatorname*{cm}(P))$,

\item ${\operatorname*{pcdp}}(P)={\operatorname*{pcdp}}(\operatorname*{cm}%
(P)).$
\end{itemize}
\end{notation}

\begin{example}
Let $P=(P_{1},P_{2},P_{3})$ be as in Example \ref{example:P-cm(P)}. Thus

\begin{itemize}
\item ${\operatorname*{dp}}(P) ={\operatorname*{dp}}({\operatorname*{cm}}(P))
={\operatorname*{dp}}\left[
\begin{array}
[c]{cccc}%
b_{03} & b_{02} & b_{01} & b_{00}\\
b_{13} & b_{12} & b_{11} & b_{10}\\
& b_{22} & b_{21} & b_{20}%
\end{array}
\right]  = c_{1}x+c_{0} $, where
\[
c_{1}=\det\left[
\begin{array}
[c]{cc|c}%
b_{03} & b_{02} & b_{01}\\
b_{13} & b_{12} & b_{11}\\
& b_{22} & b_{21}%
\end{array}
\right]  ,\ \ c_{0}={\det}\left[
\begin{array}
[c]{cc|c}%
b_{03} & b_{02} & b_{00}\\
b_{13} & b_{12} & b_{10}\\
& b_{22} & b_{20}%
\end{array}
\right]
\]

\item ${\operatorname*{pcdp}}(P)={\operatorname*{pcdp}}(\operatorname*{cm}%
(P))=c_{1}=\det\left[
\begin{array}
[c]{cc|c}%
b_{03} & b_{02} & b_{01}\\
b_{13} & b_{12} & b_{11}\\
& b_{22} & b_{21}%
\end{array}
\right]  $
\end{itemize}
\end{example}

\medskip

\noindent Next we recall the concept of subresultant for two univariate polynomials.

\begin{definition}
\label{def:sres_2polys} Let $F_{0},F_{1}\in\mathcal{Z}[x]$ with $\deg
(F_{i})=d_{i}$ and $d_{0}\leq d_{1}$. Let $0<k\leq d_{0}$.

\begin{itemize}
\item The $k$-\emph{subresultant} of $F_{0}$ and $F_{1}$, written as
$R_{k}(F_{0},F_{1})$, is defined by
\[
R_{k}(F_{0},F_{1})={\operatorname*{dp}}(x^{d_{1}-(d_{0}-k)-1}F_{0}
,\ldots,x^{0}F_{0},x^{k-1}F_{1},\ldots,x^{0}F_{1})
\]

\item The \emph{principal} \emph{coefficient} of $R_{k}(F_{0},F_{1})$, written
as $r_{k}(F_{0},F_{1})$, is defined by
\[
r_{k}(F_{0},F_{1})=\operatorname*{coeff}\nolimits_{x^{d_{0}-k}}\left(
R_{k}\left(  F_{0},F_{1}\right)  \right)
\]

\end{itemize}

\noindent One can extend the above definition to the case when $k=0$. In this
case, it is required that $d_{0}\neq d_{1}$. Then
\[
R_{0}(F_{0},F_{1})=a_{0d_{0}}^{d_{1}-d_{0}-1}F_{0}
\]
where $a_{0d_{0}}$ is the leading coefficient of $F_{0}$.
\end{definition}

\begin{example}
Let
\begin{align*}
F_{0}  &  =a_{03}x^{3}+a_{02}x^{2}+a_{01}x+a_{00}\\
F_{1}  &  =a_{14}x^{4}+a_{13}x^{3}+a_{12}x^{2}+a_{11}x+a_{10}%
\end{align*}
Let $k=2$. Then
\[
P=\left(  x^{2}F_{0},x^{1}F_{0},x^{0}F_{0},x^{1}F_{1},x^{0}F_{1}\right)
\]
Thus

\begin{itemize}
\item The $2$-\emph{subresultant} of $F_{0}$ and $F_{1},$ written as
$R_{2}(F_{0},F_{1})$, is
\begin{align*}
R_{2}(F_{0},F_{1})  &  ={\operatorname*{dp}}(P)\\
&  ={\operatorname*{dp}}({\operatorname*{cm}}(P))\\
&  ={\operatorname*{dp}}\left[
\begin{array}
[c]{cccccc}%
a_{03} & a_{02} & a_{01} & a_{00} &  & \\
& a_{03} & a_{02} & a_{01} & a_{00} & \\
&  & a_{03} & a_{02} & a_{01} & a_{00}\\\hline
a_{14} & a_{13} & a_{12} & a_{11} & a_{10} & \\
& a_{14} & a_{13} & a_{12} & a_{11} & a_{10}%
\end{array}
\right] \\
&  =\det\left[
\begin{array}
[c]{ccccc}%
a_{03} & a_{02} & a_{01} & a_{00} & \\
& a_{03} & a_{02} & a_{01} & a_{00}\\
&  & a_{03} & a_{02} & a_{01}\\\hline
a_{14} & a_{13} & a_{12} & a_{11} & a_{10}\\
& a_{14} & a_{13} & a_{12} & a_{11}%
\end{array}
\right]  x+{\det}\left[
\begin{array}
[c]{ccccc}%
a_{03} & a_{02} & a_{01} & a_{00} & \\
& a_{03} & a_{02} & a_{01} & \\
&  & a_{03} & a_{02} & a_{00}\\\hline
a_{14} & a_{13} & a_{12} & a_{11} & \\
& a_{14} & a_{13} & a_{12} & a_{10}%
\end{array}
\right]
\end{align*}

\item The \emph{principal} \emph{coefficient} of $R_{2}(F_{0},F_{1})$, written
as $r_{2}(F_{0},F_{1})$, is
\begin{align*}
r_{2}(F_{0},F_{1})  &  =\operatorname*{coeff}\nolimits_{x^{1}}\left(
R_{2}\left(  F_{0},F_{1}\right)  \right) \\
&  =\det\left[
\begin{array}
[c]{ccccc}%
a_{03} & a_{02} & a_{01} & a_{00} & \\
& a_{03} & a_{02} & a_{01} & a_{00}\\
&  & a_{03} & a_{02} & a_{01}\\\hline
a_{14} & a_{13} & a_{12} & a_{11} & a_{10}\\
& a_{14} & a_{13} & a_{12} & a_{11}%
\end{array}
\right]
\end{align*}

\end{itemize}
\end{example}

\begin{remark}
\ 

\begin{itemize}
\item Note that we are using a new indexing for subresultants. For instance,
$R_{k}(F_{0},F_{1})$ in the new indexing would have been indexed as
$R_{d_{0}-k}(F_{0},F_{1})$ in the classical indexing.

\item Note that we are using the terminologies \textquotedblleft subresultant"
and \textquotedblleft principal coefficient" as convention. In other
literatures, readers may also see the terminologies \textquotedblleft
subresultant" and \textquotedblleft principal subresultant coefficient" (e.g.,
\cite{Collins:1976a,1993_Mishra}), or \textquotedblleft polynomial
subresultant" and \textquotedblleft scalar subresultant" (e.g.,
\cite{2003_von_zur_Gathen_Lucking}), \textquotedblleft subresultant
polynomial" and \textquotedblleft subresultant" (e.g., \cite{2006_DAndrea_Krick_Szanto})
as their alternatives.
\end{itemize}
\end{remark}

\subsection{Sylvester's theory for two polynomials}

Sylvester~\cite{1853_Sylvester} gave an elegant relationship between
subresultants and the gcd of two given polynomials.

\begin{theorem}
[Sylvester's Theorem]Given $F_{0},F_{1}\in\mathcal{Z}[x]$ with $\deg
(F_{i})=d_{i}$. Then the following claims are equivalent:

\begin{itemize}
\item $\deg\gcd(F_{0},F_{1})=i$;

\item $r_{d_{0}}(F_{0},F_{1})=\cdots=r_{d_{0}-i+1}(F_{0},F_{1})=0$ and
$r_{d_{0}-i}(F_{0},F_{1})\ne0$.
\end{itemize}

\noindent Furthermore, we have:\;\;\; $\deg\gcd(F_{0},F_{1}%
)=i\ \ \ \Longleftrightarrow\ \ \ \gcd(F_{0},F_{1})=R_{d_{0}-i}(F_{0},F_{1}).
$
\end{theorem}

\subsection{Habicht's theory for two polynomials}

Habicht~\cite{Habicht:48} gave two intrinsic relationships among subresultants
of two polynomials. The first is the similarity of a subresultant with the
pseudo-remainder of its two consecutive subresultants and the second is the
similarity of a subresultant with the subresultant of two others. In this
paper, we will generalize the first result to several polynomials. Therefore,
we reproduce the first result below.

\begin{theorem}
[Habicht's theorem]\label{thm:structure_2poly} Given $F_{0}$ and $F_{1}$ of
degrees $d_{0}$ and $d_{1}=d_{0}+1$, with formal coefficients, let
$R_{-1}:=F_{1}$, $r_{-1}:=1$. Then for $k=0,\ldots,d_{0}-1$,
\[
r_{k-1}^{2}R_{k+1}=\mathrm{prem}(R_{k-1},R_{k})
\]

\end{theorem}

\section{Generalization of subresultants to several polynomials}

\label{sec:gen_sres}

In this section, we describe a natural and obvious extension of subresultant
of two polynomials to several polynomials. The main objectives are two folds:
(1) fix notations and notions and (2) review some relevant previous results.

There are several different formulation of the subresultants of two
polynomials such as Sylvester-type~\cite{1853_Sylvester},
Bezout-type~\cite{white:09,2004_WANG_HOU}, and
Barnett-type~\cite{1971_BARNETT}. In this paper, we choose the Sylvester-type
for generalization because the Sylvester form has a nice single-determinant
structure and its entries are merely the coefficients of given polynomials,
which is convenient for theory development.

\begin{notation}
\ 

\begin{itemize}

\item $P(d_{0},n)=\{(\delta_{1},\ldots,\delta_{n})\in\mathbb{N}^{n}%
:\,|\delta|=\delta_{1}+\cdots+\delta_{n}\leq d_{0}\}$;

\item $\mathcal{F}_{k}=x^{\delta_{k}-1}F_{k},\ldots,x^{0}F_{k}$ where
$\delta_{k}\in\mathbb{N}$;

\item $c\left(  \delta\right)  =\#\operatorname{col}\operatorname{cm}\left(
\mathcal{F}_{1},\ldots,\mathcal{F}_{n}\right)  . $
\end{itemize}
\end{notation}

\begin{example}
Let $d=(3,3,4)$. Then
\[
P(d_{0}%
,n)=\{\ (3,0),\ (2,1),\ (1,2),\ (0,3),\ (2,0),\ (1,1),\ (0,2),\ (1,0),\ (0,1),\ (0,0)\ \}
\]
Choose $\delta=(1,1)\in P(3,2)$. We have
\[
\mathcal{F}_{1}=x^{0}F_{1},\quad\mathcal{F}_{2}=x^{0}F_{2}
\]
Thus
\begin{align*}
\operatorname{cm}\left(  \mathcal{F}_{1},\mathcal{F}_{2}\right)   &
=\operatorname{cm}\ (x^{0}F_{1},x^{0}F_{2}) =%
\begin{bmatrix}
& a_{13} & a_{12} & a_{11} & a_{10}\\
a_{24} & a_{23} & a_{22} & a_{21} & a_{20}%
\end{bmatrix}
\\
c\left(  \delta\right)   &  =\#\operatorname{col}%
\begin{bmatrix}
& a_{13} & a_{12} & a_{11} & a_{10}\\
a_{24} & a_{23} & a_{22} & a_{21} & a_{20}%
\end{bmatrix}
=5
\end{align*}

\end{example}

\begin{definition}
[Subresultant]\label{def:sres_npoly} Let $\delta\in P(d_{0},n)$.

\begin{itemize}
\item The $\delta$-\emph{subresultant} of $F$, written as $R_{\delta}\left(
F\right)  $, is defined by
\[
R_{\delta}\left(  F\right)  ={\operatorname*{dp}}\operatorname{cm}\left(
\mathcal{F}_{0},\ldots,\mathcal{F}_{n}\right)
\]
where again%
\begin{equation}
\label{eqs:delta0}\delta_{0}=\left\{
\begin{array}
[c]{ll}%
c\left(  \delta\right)  -d_{0} & \text{if\ \ }c\left(  \delta\right)  \geq
d_{0}\\
1 & \text{else}%
\end{array}
\right.
\end{equation}

\item The \emph{principal} \emph{coefficient} of $R_{\delta}\left(  F\right)
$, written as $r_{\delta}(F)$, is defined by
\[
r_{\delta}(F)=\operatorname*{coeff}\nolimits_{x^{d_{0}-\left\vert
\delta\right\vert }}\left(  R_{\delta}\left(  F\right)  \right)
\]

\end{itemize}
\end{definition}

\begin{remark}
In the above, the particular expression for $\delta_{0}$ is chosen because it
naturally extends the formulation of subresultants for two polynomials.
Roughly speaking, such choice of $\delta_{0}$ makes the submatrix of
$\operatorname{cm}\left(  \mathcal{F}_{0},\ldots,\mathcal{F}_{n}\right)  $
involving the coefficients of $F_{0}$ the widest block while keeping the size
of the matrix $\operatorname{cm}\left(  \mathcal{F}_{0},\ldots,\mathcal{F}%
_{n}\right)  $ as small as possible.
\end{remark}

\begin{example}
Let $d=(3,3,4)$ and $\delta=\left(  1,1\right)  $. Note%
\[
c\left(  \delta\right)  =\#\operatorname{col}\operatorname{cm}\left(
\mathcal{F}_{1},\mathcal{F}_{2}\right)  =5\geq3=d_{0}\ \text{and}%
\ \ \delta_{0}=5-3=2
\]
Therefore
\begin{align*}
R_{(1,1)}(F) &  ={\operatorname*{dp}}%
\begin{bmatrix}
a_{03} & a_{02} & a_{01} & a_{00} & \\
& a_{03} & a_{02} & a_{01} & a_{00}\\
& a_{13} & a_{12} & a_{11} & a_{10}\\
a_{24} & a_{23} & a_{22} & a_{21} & a_{20}%
\end{bmatrix}
=\det%
\begin{bmatrix}
a_{03} & a_{02} & a_{01} & a_{00}\\
& a_{03} & a_{02} & a_{01}\\
& a_{13} & a_{12} & a_{11}\\
a_{24} & a_{23} & a_{22} & a_{21}%
\end{bmatrix}
x+\det%
\begin{bmatrix}
a_{03} & a_{02} & a_{01} & \\
& a_{03} & a_{02} & a_{00}\\
& a_{13} & a_{12} & a_{10}\\
a_{24} & a_{23} & a_{22} & a_{20}%
\end{bmatrix}
\\
r_{(1,1)}(F) &  =\det%
\begin{bmatrix}
a_{03} & a_{02} & a_{01} & a_{00}\\
& a_{03} & a_{02} & a_{01}\\
& a_{13} & a_{12} & a_{11}\\
a_{24} & a_{23} & a_{22} & a_{21}%
\end{bmatrix}
\end{align*}

\end{example}

\section{Generalization of Sylvester's theory to several polynomials}

The classical Sylvester's theory establishes an elegant relationship between
subresultants and the gcd of two given polynomials, yielding an elegant
algorithm. In this section, we \emph{generalize Sylvester's theory to several
polynomials} and obtain an elegant relationship between generalized
subresultants and the gcd of several polynomials, yielding an \emph{elegant
algorithm}.

\subsection{Main result}

\begin{definition}
[Graded lexicographic order]Let $\delta,\gamma\in\mathbb{N}^{n}$. We say
$\delta\succ_{\mathrm{glex}}\gamma$ if one of the two occurs:

\begin{itemize}
\item $|\delta|\;>\;|\gamma|$, or

\item $|\delta|\;=\;|\gamma|$ and the first non-zero entry of $\delta-\gamma$
is positive.
\end{itemize}
\end{definition}

\begin{notation}
$P(d_{0},n)=[(\delta_{1},\ldots,\delta_{n}):\,\delta_{1}+\cdots+\delta_{n}\le
d_{0}]$ where the tuples are decreasing under $\succ_{\mathrm{glex}}$.
\end{notation}


\begin{theorem}
[Generalized Sylvester's Theorem]\label{thm:gsylvester} Let $d\in
\mathbb{N}^{n+1}$ where $d_{0}=\min d$. Let $P(d_{0},n)=[\delta^{1}, \ldots,
\delta^{p}]$. Then we have
\[
\gcd\left(  F\right)  =\left\{
\begin{array}
[c]{lll}%
R_{\delta^{1}}\left(  F\right)  & \text{if} & r_{{\delta^{1}}}\left(
F\right)  \neq0\\
R_{{\delta^{2}}}\left(  F\right)  & \text{else if} & r_{{\delta^{2}}}\left(
F\right)  \neq0\\
\ \ \ \vdots & \ \ \ \vdots & \ \ \ \vdots\\
R_{{\delta^{p}}}\left(  F\right)  & \text{else if} & r_{{\delta^{p}}}\left(
F\right)  \neq0
\end{array}
\right.
\]

\end{theorem}

\begin{remark}
Note $\delta^{p}=(0,\ldots,0)$, in turn $R_{\delta^{p}}(F)=F_{0}$ and
$r_{\delta^{p}}\left(  F\right)  =a_{0d_{0}}$. As the result, $r_{{\delta^{p}%
}}\left(  F\right)  $ is always non-zero. Hence the last condition
$r_{{\delta^{p}}}\left(  F\right)  \neq0$ is vacuously true and thus can be omitted.
\end{remark}

\begin{example}
Let $n=1$ and $d=\left(  d_{0},d_{1}\right)  $. Then $P(d_{0},1)=[d_{0}%
,d_{0}-1,\ldots,1,0]$. By the generalized Sylvester's theorem (Theorem
\ref{thm:gsylvester}), we have
\[
\gcd\left(  F\right)  =\left\{
\begin{array}
[c]{lll}%
R_{d_{0}}\left(  F\right)  & \text{if} & r_{d_{0}}\left(  F\right)  \neq0\\
R_{d_{0}-1}\left(  F\right)  & \text{else if} & r_{d_{0}-1}\left(  F\right)
\neq0\\
\ \ \ \vdots & \ \vdots & \ \ \ \vdots\\
R_{{1}}\left(  F\right)  & \text{else if} & r_{{1}}\left(  F\right)  \neq0\\
R_{{0}}\left(  F\right)  & \text{else} &
\end{array}
\right.
\]
which is exactly the classical Sylvester's theorem, with the observation that
$\deg R_{k}(F)=d_{0}-k$.
\end{example}

\begin{example}
\label{ex:gsylvester} Let $d=(3,3,4)$. By the generalized Sylvester's theorem
(Theorem \ref{thm:gsylvester}), we have
\[
\gcd\left(  F\right)  =\left\{
\begin{array}
[c]{lll}%
R_{{(3,0)}}\left(  F\right)  & \text{if} & r_{{(3,0)}}\left(  F\right)
\neq0\\
R_{{(2,1)}}\left(  F\right)  & \text{else if} & r_{{(2,1)}}\left(  F\right)
\neq0\\
R_{{(1,2)}}\left(  F\right)  & \text{else if} & r_{{(1,2)}}\left(  F\right)
\neq0\\
R_{{(0,3)}}\left(  F\right)  & \text{else if} & r_{{(0,3)}}\left(  F\right)
\neq0\\
R_{{(2,0)}}\left(  F\right)  & \text{else if} & r_{{(2,0)}}\left(  F\right)
\neq0\\
R_{{(1,1)}}\left(  F\right)  & \text{else if} & r_{{(1,1)}}\left(  F\right)
\neq0\\
R_{{(0,2)}}\left(  F\right)  & \text{else if} & r_{{(0,2)}}\left(  F\right)
\neq0\\
R_{{(1,0)}}\left(  F\right)  & \text{else if} & r_{{(1,0)}}\left(  F\right)
\neq0\\
R_{{(0,1)}}\left(  F\right)  & \text{else if} & r_{{(0,1)}}\left(  F\right)
\neq0\\
R_{{(0,0)}}\left(  F\right)  & \text{else} &
\end{array}
\right.
\]

\end{example}

\begin{remark}
Note that
\[
\deg R_{\delta}(F)=c(\delta)-(|\delta|+\delta_{0})=(d_{0}+\delta_{0}%
)-(|\delta|+\delta_{0})=d_{0}-|\delta|.
\]
Thus the parametric gcd of a given degree could come in several different
expressions. For instance, in Example \ref{ex:gsylvester}, the parametric gcd of
degree $1$ comes in three different expressions: $R_{{(2,0)}}\left(  F\right)
,\ R_{{(}1,{1)}}\left(  F\right)  \ \,$and $R_{{(0,2)}}\left(  F\right)  $.
Roughly put, different expressions correspond to different \textquotedblleft
paths\textquotedblright\ toward degree $1$, that is,

\begin{itemize}
\item if $\ \ \deg(F_{0},F_{1})=1$ \ and $\ \deg(F_{0},F_{1},F_{2})=1$,
\ \ then $\ \ \gcd(F)\ =\ R_{{(2,0)}}\left(  F\right)  $;

\item if $\ \ \deg(F_{0},F_{1})=2$ \ and $\ \deg(F_{0},F_{1},F_{2})=1$,
\ \ then $\ \ \gcd(F)\ =\ R_{{(1,1)}}\left(  F\right)  $;

\item if $\ \ \deg(F_{0},F_{1})=3$ \ and $\ \deg(F_{0},F_{1},F_{2})=1$,
\ \ then $\ \ \gcd(F)\ =\ R_{{(0,2)}}\left(  F\right)  $.
\end{itemize}

\noindent It naturally raises a question: Can some of these expressions (and the
corresponding cells) be combined
into one? Answering that question is beyond the scope of this paper, and
thus is left for the future research.
\end{remark}

\subsection{Proof of the generalized Sylvester's theorem
(Theorem~\ref{thm:gsylvester})}

This subsection is devoted to proving the generalized Sylvester's theorem. The
proof is divided into three steps. First, we show that a subresultant can be
written as the combination of the input polynomials with the leading
coefficients of the co-factors explicitly given. Then we decompose the result
of the generalized Sylvester's theorem into two lemmas and show their
correctness, respectively. Based on the two lemmas, we provide a proof of the
generalized Sylvester's theorem.

\medskip

\begin{notation}
We begin by introducing a few notations.

\begin{itemize}
\item $\Delta_{d}:=\{\mu=(\mu_{0},\ldots,\mu_{n})\in\mathbb{N}^{n+1}%
\backslash\{\boldsymbol{0}\}: \mu_{0}>0,\ |\mu|\le\max\limits_{\substack{i\ge
0\\\delta_{i}\ne0}}(d_{i}+\delta_{i})=d_{0}+\mu_{0}\}$

\item For $\mu\in\Delta_{d}$, $M_{\mu}(F):={\operatorname*{cm}}(x^{\mu_{0}%
-1}F_{0},\ldots,x^{0}F_{0},\ldots\ldots,x^{\mu_{n}-1}F_{n},\ldots,x^{0}F_{n})$;

\item For $\mu\in\Delta_{d}$, $\widetilde{R}_{\mu}(F):=\operatorname{dp}
M_{\mu}(F)$ and $\widetilde{r}_{\mu}(F)$ is the principal coefficient of
$\widetilde{R}_{\mu}(F)$.
\end{itemize}
\end{notation}

Below, we will show that $\widetilde{R}_{\mu}$ can be written as a combination
of the input polynomials, which is expected. Furthermore, we provide an
explicit expression for the leading coefficients of the co-factors of the
input polynomials in the combination. For the sake of simplicity, when $F$ is
clear from the context, we abbreviate~$\widetilde{R}_{\mu}(F)$ and
$\widetilde{r}_{\mu}(F)$ as $\widetilde{R}_{\mu}$ and $\widetilde{r}_{\mu}$, respectively.

\begin{lemma}
\label{lem:sres_multipoly} Let $\mu\in\Delta_{d}$. Then there exist
$H_{0},\ldots,H_{n}$ such that

\begin{enumerate}
\item $\widetilde{R}_{\mu}(F)=H_{0}F_{0}+\cdots+H_{n}F_{n}$,

\item $\deg H_{i}\leq\mu_{i} -1$, and

\item for $i\ge1$, $\operatorname*{coeff}(H_{i},x^{\mu_{i}-1}) =
(-1)^{\sigma_{i}}\widetilde{r}_{\tilde{\mu}}$
where

\begin{enumerate}
\item $\sigma_{i}=1+\mu_{i}+\cdots+\mu_{n}$

\item $\tilde{\mu}=(\mu_{0},\ldots,\mu_{i-1},\mu_{i}-1,\mu_{i+1},\ldots
,\mu_{n})$.
\end{enumerate}
\end{enumerate}
\end{lemma}

\begin{proof}
Let $\tilde{d}_{i}=\max_{\substack{k\ge0 \\\mu_{k}\ne0}}(d_{k}+\mu_{k}%
)-\mu_{i}$ for $i=0,\ldots,n$. Then
\[
M_{\mu}=\left[
\begin{array}
[c]{cccccc}%
a_{0\tilde{d}_{0}} & \cdots & a_{01} & a_{00} &  & \\
& \ddots &  & \ddots & \ddots & \\
&  & a_{0\tilde{d}_{0}} & \cdots & a_{01} & a_{00}\\\hline
a_{1\tilde{d}_{1}} & \cdots & a_{11} & a_{10} &  & \\
& \ddots &  & \ddots & \ddots & \\
&  & a_{1\tilde{d}_{1}} & \cdots & a_{11} & a_{10}\\\hline
\vdots & \vdots & \vdots & \vdots & \vdots & \vdots\\
\vdots & \vdots & \vdots & \vdots & \vdots & \vdots\\\hline
a_{n\tilde{d}_{n}} & \cdots & a_{n1} & a_{n0} &  & \\
& \ddots &  & \ddots & \ddots & \\
&  & a_{n\tilde{d}_{n}} & \cdots\cdots & a_{n1} & a_{n0}%
\end{array}
\right]
\]
where $a_{ij}=0$ when $j>d_{i}$. By applying the multi-linearity of
determinants, we have
\[
\widetilde{R}_{\mu}(F)=\operatorname*{dp}M_{\mu}=\det%
\begin{array}
[c]{l}%
\left[
\begin{array}
[c]{ccccc}%
a_{0\tilde{d}_{0}} & \cdots & \cdots & \cdots & x^{\mu_{0}-1}F_{0}\\
& \ddots & \ddots &  & \vdots\\
&  & a_{0\tilde{d}_{0}} & \cdots & x^{0}F_{0}\\\hline
\vdots & \vdots & \vdots & \vdots & \vdots\\\hline
a_{i\tilde{d}_{i}} & \cdots & \cdots & \cdots & x^{\mu_{i}-1}F_{i}\\
& \ddots & \ddots &  & \vdots\\
&  & a_{i\tilde{d}_{i}} & \cdots & x^{0}F_{i}\\\hline
\vdots & \vdots & \vdots & \vdots & \vdots\\\hline
a_{n\tilde{d}_{n}} & \cdots & \cdots & \cdots & x^{\mu_{n}-1}F_{n}\\
& \ddots & \ddots &  & \vdots\\
&  & a_{n\tilde{d}_{n}} & \cdots & x^{0}F_{n}%
\end{array}
\right] \\[-8pt]%
~~\underbrace{\hspace{10em}}_{\text{The first }|\mu|-1\text{ columns of
}M_{\mu}}%
\end{array}
\]
Taking the Laplace expansion on the last column, we have
\[
\widetilde{R}_{\mu}(F)=\sum_{i=0}^{n}\det%
\begin{array}
[c]{l}%
\left[
\begin{array}
[c]{ccccc}%
a_{0\tilde{d}_{}} & \cdots & \cdots & \cdots & 0\\
& \ddots & \ddots &  & \vdots\\
&  & a_{0\tilde{d}_{0}} & \cdots & 0\\\hline
\vdots & \vdots & \vdots & \vdots & \vdots\\\hline
a_{i\tilde{d}_{i}} & \cdots & \cdots & \cdots & x^{\mu_{i}-1}F_{i}\\
& \ddots & \ddots &  & \vdots\\
&  & a_{i\tilde{d}_{i}} & \cdots & x^{0}F_{i}\\\hline
\vdots & \vdots & \vdots & \vdots & \vdots\\\hline
a_{n\tilde{d}_{n}} & \cdots & \cdots & \cdots & 0\\
& \ddots & \ddots &  & \vdots\\
&  & a_{n\tilde{d}_{n}} & \cdots & 0
\end{array}
\right] \\[-8pt]%
~~\underbrace{\hspace{10em}}_{\text{The first }|\mu|-1\text{ columns of
}M_{\mu}}%
\end{array}
=\sum_{i=0}^{n}\det\left[
\begin{array}
[c]{ccccc}%
a_{0\tilde{d}_{0}} & \cdots & \cdots & \cdots & 0\\
& \ddots & \ddots &  & \vdots\\
&  & a_{0\tilde{d}_{0}} & \cdots & 0\\\hline
\vdots & \vdots & \vdots & \vdots & \vdots\\\hline
a_{i\tilde{d}_{i}} & \cdots & \cdots & \cdots & x^{\mu_{i}-1}\\
& \ddots & \ddots &  & \vdots\\
&  & a_{i\tilde{d}_{i}} & \cdots & x^{0}\\\hline
\vdots & \vdots & \vdots & \vdots & \vdots\\\hline
a_{n\tilde{d}_{n}} & \cdots & \cdots & \cdots & 0\\
& \ddots & \ddots &  & \vdots\\
&  & a_{n\tilde{d}_{n}} & \cdots & 0
\end{array}
\right]  F_{i}%
\]
Denote the coefficient of $F_{i}$ in the right-hand side of the above equation
by $H_{i}$ for $i=0,\ldots,n$. Then $\widetilde{R}_{\mu}=H_{0}F_{0}%
+\cdots+H_{n}F_{n}$, which implies $\widetilde{R}_{\mu}(F)\in\langle F\rangle
$. It is easy to see that $\deg H_{i}<\mu_{i}$.

Assume $H_{i}=\sum_{j=0}^{\mu_{i}-1}c_{ij}x^{j}$. When $i\ge1$, the principal
coefficient of $H_{i}$, i.e., the coefficient of $H_{i}$ in the term
$x^{\mu_{i}-1}$, is given by
\begin{equation}
c_{i(\mu_{i}-1)}=\det\left[
\begin{array}
[c]{cccccc}%
a_{0\tilde{d}_{0}} & \cdots & \cdots & \cdots & \cdots & 0\\
& \ddots & \ddots &  &  & \vdots\\
&  & a_{0\tilde{d}_{0}} & \cdots & \cdots & 0\\\hline
\vdots & \vdots & \vdots & \vdots & \vdots & \vdots\\\hline
a_{i\tilde{d}_{i}} & \cdots & \cdots & \cdots & \cdots & 1\\
& a_{i\tilde{d}_{i}} & \cdots & \cdots & \cdots & 0\\
&  & \ddots &  &  & \vdots\\
&  &  & a_{i\tilde{d}_{i}} & \cdots & 0\\\hline
\vdots & \vdots & \vdots & \vdots & \vdots & \vdots\\\hline
a_{n\tilde{d}_{n}} & \cdots & \cdots & \cdots & \cdots & 0\\
& \ddots & \ddots &  &  & \vdots\\
&  & a_{n\tilde{d}_{n}} & \cdots & \cdots & 0
\end{array}
\right]  =(-1)^{\sigma_{i}}\widetilde{r}_{\tilde{\mu}} \label{eqs:lc_ci}%
\end{equation}
where
\begin{enumerate}
[(a)]

\item $\sigma_{i}=1+\mu_{i}+\cdots+\mu_{n}$

\item $\tilde{\mu}=(\mu_{1},\ldots,\mu_{i-1},\mu_{i}-1,\mu_{i+1},\ldots
,\mu_{n})$.\vspace{-2em}
\end{enumerate}
\end{proof}

\smallskip

\begin{remark}
For $\delta\in P(d_{0},n)$ with $c(\delta)>d_{0}$, let $\overline{\delta}$ be
such that $R_{\delta}=\widetilde{R}_{\overline{\delta}}$, that is,
$\overline{\delta}=(\delta_{0},\delta_{1},\ldots,\delta_{n})$ where
$\delta_{0}$ is chosen as in \eqref{eqs:delta0}. It is easy to see that
$\overline{\delta}\in\Delta_{d}$. Thus there exist $H_{0},\ldots,H_{n}$ such that

\begin{enumerate}
\item $R_{\delta}(F)=H_{0}F_{0}+\cdots+H_{n}F_{n}$,

\item $\deg H_{i}\leq\delta_{i} -1$, and

\item for $i\ge1$, $\operatorname*{coeff}(H_{i},x^{\delta_{i}-1}) =
(-1)^{\sigma_{i}}\widetilde{r}_{\tilde{\delta}}$
where

\begin{enumerate}
\item $\sigma_{i}=1+\delta_{i}+\cdots+\delta_{n}$

\item $\tilde{\delta}=(\delta_{0},\ldots,\delta_{i-1},\delta_{i}%
-1,\delta_{i+1},\ldots,\delta_{n})$.
\end{enumerate}
\end{enumerate}
\end{remark}

The following theorem lists several relevant results on the generalized
Sylvester matrix. They were first presented by Kakie in \cite{1976_Kakie} and
further developed by Ho in \cite{Ho:89}.

\begin{theorem}
[Kakie \cite{1976_Kakie} and Ho \cite{Ho:89}.]\label{thm:Kakie_Ho} Let
$h=\deg\gcd F$ and $N_{k}\left(  F\right)  =M_{\mu}\left(  F\right)  $ where%
\[
\mu_{i}=m+\ell-d_{i}-k,\ \ \ \ \ \ \ \ m=\max\limits_{0\leq i\leq n}%
d_{i},\ \ \ \ \ \ell=\min\limits_{0\leq i\leq n}d_{i}%
\]
Then we have

\begin{enumerate}
\item \label{lem:rank_decrease}$\mathrm{rank}\,N_{k}(F)=\mathrm{rank}%
\,N_{0}(F)-k$ for $k\leq h$.\footnote[1]{It is allowed that $k<0$.}

\item \label{lem:rank_R0}$\mathrm{rank}\,N_{0}(F)=m+\ell-h$.

\item \label{lem:gcd}Suppose $U_{1},\ldots,U_{m+\ell-h}$ are the $m+\ell-h$
linearly independent row vectors of $N_{0}(F)$. Then
\[
\gcd F=\mathrm{dp}%
\begin{bmatrix}
U_{1}\\
\vdots\\
U_{m+\ell-h}%
\end{bmatrix}
\]
where the difference by a constant factor is allowed.

\item \label{lem:ind_rows}Denote the $i$-th block of $N_{0}(F)$ by $N_{0i}$
and the $j$-th row of $N_{0i}$ by $N_{0i,j}$ where $i=0,1,\ldots,n$. If
$N_{0i,j}$ can be written as a linear combination of the rows in
$N_{01},\ldots,N_{0(i-1)}$ and $N_{0i,j+1},\ldots,N_{0i,m+\ell-d_{i}}$, then
so is $N_{0i,j-1}$.
\end{enumerate}
\end{theorem}

For proving the generalized Sylvester's theorem (Theorem \ref{thm:gsylvester}%
), we identify two essential ingredients it consists of, and reformulate them
as the following two lemmas (i.e., Lemma \ref{lem:wR_theta} and Lemma
\ref{lem:vanishing_wR_delta}) whose correctness will be shown below.

\begin{notation}
\label{notation:gcd} \ 

\begin{itemize}
\item $G_{i}=\gcd(F_{0},\ldots,F_{i})$ with $G_{0}:=F_{0}$;

\item $e_{i}=\deg G_{i}$;

\item ${\theta}=(\theta_{1},\ldots,\theta_{n})$ with $\theta_{i}=e_{i-1}%
-e_{i}$ for $i=1,\ldots,n$.
\end{itemize}
\end{notation}

\begin{remark}
It is easy to see that $\theta\in P(d_{0},n)$.
\end{remark}

\begin{lemma}
\label{lem:wR_theta} $R_{{\theta}}$ is similar to $\gcd F$.
\end{lemma}

\begin{proof}
Let $\overline{\theta}=(\theta_{0},\theta_{1},\ldots,\theta_{n})$ where
\[
\theta_{0}=\left\{
\begin{array}
[c]{ll}%
c(\theta)-d_{0} & \text{if}\ \exists_{i>0}\theta_{i}\neq0;\\
1 & \text{else}%
\end{array}
\right.
\]
Then $R_{\theta}=\widetilde{R}_{\overline{\theta}}$. Recall $N_{0}\left(
F\right)  =M_{\mu}\left(  F\right)  $ where $\mu_{i}=m+\ell-d_{i}%
,\ m=\max\limits_{0\leq i\leq n}d_{i}$ and $\ell=\min\limits_{0\leq i\leq
n}d_{i}.$ By Theorem~\ref{thm:Kakie_Ho}-3, we only need to show that the
determinant polynomial of the matrix consisting of a maximal linearly
independent subset of the rows of $N_{0}(F)$ is similar to $\text{dp}%
\ M_{\overline{\theta}}$. In what follows, we will find such a maximal
linearly independent subset of rows in a recursive way.

Let $m_{i}=\max_{0\le j\le i}d_{j}$. One may observe that
\[
N_{0}(F)=\left[
\begin{array}
[c]{cccccc}%
a_{0d_{0}} & \cdots\cdots & a_{01} & a_{00} &  & \\
& \ddots &  & \ddots & \ddots & \\
&  & a_{0d_{0}} & \cdots & a_{01} & a_{00}\\\hline
\vdots & \vdots & \vdots & \vdots & \vdots & \vdots\\
\vdots & \vdots & \vdots & \vdots & \vdots & \vdots\\\hline
a_{nd_{n}} & \cdots & a_{n1} & a_{n0} &  & \\
& \ddots &  & \ddots & \ddots & \\
&  & a_{nd_{n}} & \cdots\cdots & a_{n1} & a_{n0}%
\end{array}
\right]  \hspace{-1.5em}%
\begin{array}
[c]{l}%
\left.
\begin{array}
[c]{l}%
\\
\\
\\
\end{array}
\right\}  m_{n}+\ell-d_{0}(=m_{n})~\text{rows}\\[15pt]%
\\[35pt]%
\left.
\begin{array}
[c]{l}%
\\
\\
\\
\end{array}
\right\}  m_{n}+\ell-d_{n}~\text{rows}%
\end{array}
\]
\noindent can be viewed as appending $m_{n}+\ell-d_{n}$ rows of the
coefficients of $F_{n}$ to the matrix
\[
\left[
\begin{array}
[c]{cccccc}%
a_{0d_{0}} & \cdots\cdots & a_{01} & a_{00} &  & \\
& \ddots &  & \ddots & \ddots & \\
&  & a_{0d_{0}} & \cdots & a_{01} & a_{00}\\\hline
\vdots & \vdots & \vdots & \vdots & \vdots & \vdots\\
\vdots & \vdots & \vdots & \vdots & \vdots & \vdots\\\hline
a_{(n-1)d_{n-1}} & \cdots & a_{(n-1)1} & a_{(n-1)0} &  & \\
& \ddots &  & \ddots & \ddots & \\
&  & a_{(n-1)d_{n-1}} & \cdots & a_{(n-1)1} & a_{(n-1)0}%
\end{array}
\right]  \hspace{-1.5em}%
\begin{array}
[c]{l}%
\left.
\begin{array}
[c]{l}%
\\
\\
\\
\end{array}
\right\}  m_{n}+\ell-d_{0}(=m_{n})~\text{rows}\\[15pt]%
\\[35pt]%
\left.
\begin{array}
[c]{l}%
\\
\\
\\
\end{array}
\right\}  m_{n}+\ell-d_{n-1}~\text{rows}%
\end{array}
\]
which is exactly $N_{-(m_{n}-m_{n-1})}(F_{0},\ldots,F_{n-1})$. Hence, the
number of linearly independent rows selected from the last block of $N_{0}(F)
$ is $\operatorname{rank}\,N_{0}(F)-\operatorname{rank}\,N_{-(m_{n}-m_{n-1}%
)}(F_{0},\ldots,F_{n-1})$ .

\medskip

\noindent By Theorem \ref{thm:Kakie_Ho}-\ref{lem:rank_decrease},
\begin{equation}
\operatorname{rank}\,N_{-(m_{n}-m_{n-1})}(F_{0},\ldots,F_{n-1}%
)=\operatorname{rank}\,N_{0}(F_{0},\ldots,F_{n-1})+(m_{n}-m_{n-1})
\label{eq:rank_R0_n-1}%
\end{equation}

\noindent Recall $G_{k}=\gcd(F_{0},\ldots,F_{k})$ and $e_{k}=\deg G_{k}$. By
Theorem \ref{thm:Kakie_Ho}-\ref{lem:rank_R0},
\begin{align}
\mathrm{rank}\,N_{0}(F)\ =\ \mathrm{rank}\,N_{0}(F_{0},\ldots,F_{n})  &
=m_{n}+\ell-e_{n}\label{eq:rank_R0_Fn}\\
\operatorname{rank}\,N_{0}(F_{0},\ldots,F_{n-1})  &  =m_{n-1}+\ell-e_{n-1}
\label{eq:rank_R0_Fn-1}%
\end{align}

\medskip

\noindent Combining \eqref{eq:rank_R0_n-1}-\eqref{eq:rank_R0_Fn-1}, we have
\[
N_{0}(F)-\operatorname{rank}\,N_{-(m_{n}-m_{n-1})}(F_{0},\ldots,F_{n-1}%
)=e_{n-1}-e_{n}%
\]
\noindent By Theorem \ref{thm:Kakie_Ho}-\ref{lem:ind_rows}, a maximal linearly
independent subset of rows for $N_{0}(F)$ consists of the last $e_{n-1}-e_{n}$
rows of the last block of $N_{0}(F)$ and a maximal linearly independent subset
of rows for $N_{-(m_{n}-m_{n-1})}(F_{0},\ldots,F_{n-1})$.

\medskip

\noindent Note that $N_{-(m_{n}-m_{n-1})}(F_{0},\ldots,F_{n-1})$ can be
obtained by appending $m_{n}-m_{n-1}$ rows of the coefficients of $F_{i}$ to
the $i$-th block of $N_{0}(F_{0},\ldots,F_{n-1})$ for $i=0,\ldots,n-1$ and its
rank is $\operatorname{rank}\,N_{0}(F_{0},\ldots,F_{n-1})+(m_{n}-m_{n-1})$.
Thus we can obtain a maximal linearly independent subset of rows for
$N_{-(m_{n}-m_{n-1})}(F_{0},\ldots,F_{n-1})$ by selecting the first
$m_{n}-m_{n-1}$ rows of the $0$-th block and a maximal linearly independent
subset of rows for $N_{0}(F_{0},\ldots,F_{n-1})$.

\medskip

\noindent From the above derivation, we see that $N_{0}(F)$ has a maximal
linearly independent subset of rows consisting~of:

\begin{itemize}
\item the last $e_{n-1}-e_{n}(=\theta_{n})$ rows of the last block,

\item the first $m_{n}-m_{n-1}$ rows of the $0$-th block, and

\item a maximal linearly independent subset of rows for $N_{0}(F_{0}%
,\ldots,F_{n-1})$.
\end{itemize}

\noindent With the same technique, we carry out the above procedure for
$N_{0}(F_{0},\ldots,F_{i})$ with $i=n-1,\ldots,1$ and finally obtain a maximal
linearly independent subset of rows for $N_{0}(F)$ which consists of:

\begin{itemize}
\item the last $e_{i-1}-e_{i}(=\theta_{i})$ rows of the $i$-th block for
$i=1,\ldots,n$,

\item the first $(m_{n}-m_{n-1})+\cdots+(m_{1}-m_{0})=m_{n}-m_{0}$ rows of the
$0$-th block, and

\item a maximal linearly independent subset of rows for $N_{0}(F_{0})$ whose
number of rows is exactly $d_{0}$.
\end{itemize}

\noindent Thus the number of rows chosen from the $0$-th block of $N_{0}(F)$
is $m_{n}-m_{0}+d_{0}=m_{n}$ (since $m_{0}=d_{0}$).

\medskip

\noindent One may count the total number of these rows which is $(e_{0}%
-e_{1})+\cdots+(e_{n-1}-e_{n})+m_{n}=e_{0}-e_{n}+m_{n}$. Utilizing the
relation $e_{0}=m_{0}=d_{0}$ and $e_{n}=\deg\gcd F$, we further simplify the
expression and obtain that the number of a maximal linearly independent subset
of rows for $N_{0}(F)$ is $m_{n}+d_{0}-\deg\gcd F$. By
Theorem~\ref{thm:Kakie_Ho}-\ref{lem:gcd}, we have $\widetilde{R}%
_{(m_{n},\theta_{1},\ldots,\theta_{n})}=\operatorname{dp}\,M_{(m_{n}%
,\theta_{1},\ldots,\theta_{n})}$ is similar to $\gcd F$. Next we show that
$\widetilde{R}_{(m_{n},\theta_{1},\ldots,\theta_{n})}$ is similar to
$\widetilde{R}_{\overline{\theta}}$ (i.e., $R_{\theta}$).

\medskip\noindent Note that
\[
m_{n}+d_{0}=\max_{k\ge0}d_{k}+d_{0}\ge\max_{k> 0}d_{k}+\max_{k> 0} \theta_{k}
\]
If there exists $i\ge1$ such that $\theta_{i}>0$, then
\[
m_{n}+d_{0}\ge\max_{k> 0}d_{k}+\max_{k> 0} \theta_{k}\geq\max_{\substack{k>0
\\\theta_{k}>0}}(\theta_{k}+d_{k})
\]
Hence $m_{n}\geq\max_{\substack{k>0 \\\theta_{k}>0}}(\theta_{k}+d_{k}%
)-d_{0}=\theta_{0}$. By the property of determinant polynomial, we have
\[
\widetilde{R}_{(m_{n},\theta_{1},\ldots,\theta_{n})}=a_{0d_{0}}^{m_{n}%
-\theta_{0}}\widetilde{R}_{(\theta_{0},\theta_{1},\ldots,\theta_{n}%
)}=a_{0d_{0}}^{m_{n}-\theta_{0}}\widetilde{R}_{\overline{\theta}}=a_{0d_{0}%
}^{m_{n}-\theta_{0}}R_{\theta}
\]
Otherwise, $\theta_{i}=0$ for $i=1,\ldots,n$ and thus
\[
\widetilde{R}_{(m_{n},0,\ldots,0)}=a_{0d_{0}}^{m_{n}-1}F_{1}=a_{0d_{0}}%
^{m_{n}-1}\widetilde{R}_{\overline{\theta}}=a_{0d_{0}}^{m_{n}-1}R_{\theta}
\]
For both cases, we have that $R_{\theta}$ is similar to $\gcd F$.
\end{proof}

\begin{lemma}
\label{lem:vanishing_wR_delta} For $\delta\in P(d_{0},n)$, if $\delta
\succ_{\mathrm{glex}}{\theta}$, then $R_{\delta}=0$.
\end{lemma}

\begin{proof}
We consider two cases of $\delta\succ_{\text{glex}}{\theta}$.

\begin{enumerate}
[{Case} 1]

\item $|\delta|>|\gamma|$.

In this case, $\delta_{1}+\cdots+\delta_{n}>\theta_{1}+\cdots+\theta_{n}$.
Assume $\delta_{0}$ and $\theta_{0}$ are chosen as in \eqref{eqs:delta0},
respectively. By Definition \ref{def:sres_npoly},
\begin{align*}
\deg R_{\delta}  &  =(d_{0}+\delta_{0})-(\delta_{0}+\cdots+\delta_{n})\\
&  <d_{0}-(\theta_{1}+\cdots+\theta_{n})\\
&  =(d_{0}+\theta_{0})-(\theta_{0}+\cdots+\theta_{n})\\
&  =\deg R_{\theta}\\
&  =\deg\gcd F
\end{align*}
By Lemma \ref{lem:sres_multipoly}, $R_{\delta}=\widetilde{R}_{\overline
{\delta}}\in\langle F\rangle=\langle\gcd F\rangle$, which immediately implies
that $R_{\delta}=0$.

\item $|\delta|=|\gamma|$ and the first non-zero entry of $\delta-\gamma$ is positive.

\begin{enumerate}
[(1)]

\item In this case, there exists $i>0$ such that $\delta_{1}=\theta
_{1}\;\wedge\;\cdots\wedge\;\delta_{i-1}=\theta_{i-1}\;\wedge\;\delta
_{i}>\theta_{i}$. Since $\delta_{1}+\cdots+\delta_{n}\le d_{0}$, we have
$\delta_{i}\le d_{0}$, which indicates that $\theta_{i}<d_{0}$. Hence from
$(\theta_{i}+1)+d_{i}\le d_{0}+\max_{0\le j\le i}d_{j}$, we can derive that
the last $(\theta_{i}+1)$-th row in the $i$-th block of $M_{\overline{\delta}%
}$ is also in the matrix $N_{-(m_{n}-m_{i})}(F_{0},\ldots,F_{i})$, which
implies that it can be written as a linear combination of any maximal linearly
independent rows of $N_{-(m_{n}-m_{i})}(F_{0},\ldots,F_{i})$.

\item From the proof of Lemma \ref{lem:wR_theta}, $N_{-(m_{n}-m_{i})}%
(F_{0},\ldots,F_{i})$ has a maximal subset of linearly independent rows
consisting of:

\begin{enumerate}
[(a)]

\item the last $e_{j-1}-e_{j}(=\theta_{j})$ rows of the $j$-th block of
$N_{0}(F_{0},\ldots,F_{i})$ for $j=1,\ldots,i$, and

\item all the $m_{i}$ rows of the $0$-th block.
\end{enumerate}

\item Let $\delta_{0}$ be chosen as in \eqref{eqs:delta0} and $\overline
{\delta}=(\delta_{0},\delta_{1},\ldots,\delta_{n})$. Note that $N_{-(m_{n}%
-m_{i})}(F_{0},\ldots,F_{i})$ is the submatrix of $N_{0}(F_{0},\ldots,F_{n})$
by selecting its first $i$ blocks. Thus the last $(\theta_{i}+1)$-th row of
the $i$-th block of $M_{\overline{\delta}}$ for $i>0$ is a linear combination
of the rows described in (2).

\item Consider $N_{0}(F)$ and $M_{\overline{\delta}}(F)$. Denote their the
blocks consisting of the coefficients of $F_{i}$ by $N_{0i}$ and
$M_{\overline{\delta},i}$, respectively. Since the number of the rows in the
block $N_{0i}$ is $d_{0}+\max_{k\ge1}{d_{k}}-d_{i}$ and the number of rows in
the block $M_{\overline{\delta},i}$ is $\delta_{i}$, one can easily verify
that $d_{0}+\max_{k\ge1}{d_{k}}-d_{i}\ge\delta_{i}$. Therefore, $M_{\overline
{\delta}}(F)$ is a submatrix of $N_{0}(F)$.

\item Furthermore, we have

\begin{itemize}
\item the $0$-th block of $M_{\overline{\delta}}(F)$ is the $0$-th block of
$N_{0}(F)$ by deleting its first $m_{n}-\delta_{0}$ rows,

\item the $j$-th block of $M_{\overline{\delta}}(F)$ for $j=1,\ldots,i-1$ is
exactly the last $\theta_{j}$ rows of the $j$-th block of $N_{0}(F)$, and

\item the $i$-th block of $M_{\overline{\delta}}(F)$ is the last $\theta_{i}$
rows of the $i$-the block of $N_{0}(F)$ appending $\delta_{i}-\theta_{i}$ rows
which are linear combinations of the first $\delta_{0}+\theta_{1}%
+\cdots+\theta_{i}$ rows of $M_{\overline{\delta}}$.
\end{itemize}

\item From (5), one immediately concludes that $R_{\delta}=\operatorname*{dp}
M_{\overline{\delta}}(F)=0$.
\end{enumerate}
\end{enumerate}
\end{proof}

\noindent Now we are ready to prove the generalized Sylvester's theorem
(Theorem \ref{thm:gsylvester}). \medskip

\begin{proof}
[Proof of the generalized Sylvester's theorem (Theorem \ref{thm:gsylvester})]
Note that the theorem is equivalent to the following claim:
\[
\gcd(F)=R_{\delta}%
\]
where
\[
\delta=\max\limits_{\substack{\gamma\in P(d_{0},n) \\r_{\gamma}\ne0}}\gamma
\]
where again $\max$ is with respect to the ordering $\succ
_{\operatorname*{glex}}$. We will show the correctness of the claim.

\medskip\noindent By Lemma \ref{lem:wR_theta}, $\gcd(F)=R_{\theta}$. Moreover,
it is easy to see that $\theta\in P(d_{0},n)$. We will show $\delta=\theta$ by
disproving {$\delta\succ_{\text{glex}}\theta$} and $\theta\succ_{\text{glex}%
}\delta$.

\medskip\noindent If $\delta\succ_{\text{glex}}\theta$, by Lemma
\ref{lem:vanishing_wR_delta}, $R_{\delta}=0$. However, it contradicts the
condition for determining $\delta$.

\medskip\noindent If $\theta\succ_{\text{glex}}\delta$, by the condition for
determining $\delta$, we immediately have $r_{\theta} =0$. Assume $\delta_{0}$
and $\theta_{0}$ are chosen as in \eqref{eqs:delta0}, respectively. It follows
that
\begin{align}
\deg R_{\theta}  &  <(\theta_{0}+d_{0})-\sum_{i=0}^{n}\theta_{i}\nonumber\\
&  =d_{0}-\sum_{i=1}^{n}\theta_{i}\nonumber\\
&  =d_{0}-\sum_{i=1}^{n}(e_{i-1}-e_{i})\nonumber\\
&  =d_{0}-(e_{0}-e_{n})\nonumber\\
&  =e_{n}. \label{eq:contradiction}%
\end{align}
On the other hand, $\gcd(F)=R_{\theta}$ implies that $e_{n}=\deg\gcd
(F_{0},\ldots, F_{n})=\deg R_{\theta}$, which is a contradiction with \eqref{eq:contradiction}.

\medskip\noindent Therefore, the only possibility is $\delta=\theta$, which
immediately deduces that $\gcd(F)=R_{\theta}$ by Lemma \ref{lem:wR_theta}.
\end{proof}

\begin{remark}
In the above, we provide an algebraic proof for a main result (Theorem
\ref{thm:gsylvester}). In the preprint~\cite{2021_Hong_Yang_b}, we also
provided a geometric proof where the subresultants are expressed in terms of
the roots of polynomials, Such approach was pioneered by Sylvester in his
major work~\cite{1853_Sylvester} and later simplified and generalized by
Bostan, D'Andrea, Hong, Krick, Szanto and Valdettaro in \cite{Hong:99a,
Hong:99b,Hong:2000a, 2006_DAndrea_Krick_Szanto, 2007_DAndrea_Hong_Krick_Szanto,
2009_DAndrea_Hong_Krick_Szanto, 2013_DAndrea_Krick_Szanto,
2015_DAndrea_Krick_Szanto, 2017_Bostan_DAndrea_Krick_Szanto_Valdettaro,
2017_Krick_Szanto_Valdettaro, 2019_Andrea_Krick_Szanto_Valdettaro}.
\end{remark}

\subsection{An algorithm for computing parametric gcd with the generalized
Sylvester's theorem}

With the help of the generalized Sylvester's theorem (Theorem
\ref{thm:gsylvester}), we immediately come up with the following elegant
algorithm for computing the parametric gcd of several univariate polynomials.

\begin{algorithm}
\label{alg:non_recursive_naive} $\mathcal{G}\leftarrow$ PGCD($d$)

\begin{itemize}
\item[In \ :] $d$ is a list of degrees

\item[Out:] $\mathcal{G}$ is a representation of the parametric gcd for $d$,
that is,%
\[
\mathcal{G}=\left(  \left(  H_{1},G_{1}\right)  ,\ldots,\left(  H_{p}%
,G_{p}\right)  \right)
\]
such that%
\[
\gcd\left(  F\right)  =\left\{
\begin{array}
[c]{lll}%
G_{1} & \text{if} & H_{1}\neq0\\
G_{2} & \text{else if} & H_{2}\neq0\\
\ \vdots & \ \ \ \vdots & \ \ \ \vdots\\
G_{p} & \text{else if} & H_{p}\neq0
\end{array}
\right.
\]

\end{itemize}

\begin{enumerate}
\item $F\leftarrow(F_{0},\ldots,F_{n})$ where $n=\#d-1$ and $F_{i}=\sum
_{j=0}^{d_{i}}a_{ij}x^{j}\ $(where $a_{ij}$ are indeterminates)

\item $P(d_{0},n)\leftarrow\left[  \ (\delta_{1},\ldots,\delta_{n}%
):\ \delta_{1}+\cdots+\delta_{n}\le d_{0}\ \right]  $ ordered decreasingly
under $\succ_{\mathrm{glex}}$

\item $R_{\delta}\leftarrow R_{\delta}\left(  F\right)  \ $for $\delta\in
P(d_{0},n)$

$r_{\delta}\leftarrow\operatorname{lc}_{x} R_{\delta}\left(  F\right)  \ $

\item $\mathcal{G}\leftarrow\left(  \left(  r_{\delta},R_{\delta}\right)
:\delta\in P(d_{0},n)\right)  $
\end{enumerate}
\end{algorithm}

\begin{proposition}
Algorithm \ref{alg:non_recursive_naive} terminates in finite steps and the
output is correct.
\end{proposition}

\begin{proof}
The termination of Algorithm \ref{alg:non_recursive_naive} is obvious because
$P(d_{0},n)$ is a finite set whose finiteness is implied by $\delta_{1}%
+\cdots+\delta_{n}\leq d_{0}$. The correctness of the algorithm is guaranteed
by the generalized Sylvester's theorem.
\end{proof}

\section{Generalization of Habicht's theory to several polynomials}

\label{sec:gen_habicht}

Habicht's theory is another key ingredient in the classical subresultant
theory. It establishes a systematic relationship between subresultants and
pseudo-remainders, yielding an efficient algorithm for computing the
subresultant chain of two polynomials. In this section, we \emph{generalize
Habicht's theory to several polynomials}, in order to obtain a systematic
relationship between generalized subresultants and pseudo-remainders (the
generalized Habicht's theorem). We then describe a resulting \emph{efficient
algorithm}. It should be pointed out that the generalized Habicht's theory
presented in this section is an analogy of the classical Habicht's theory. It
adopts the same idea of converting the computation of subresultants into that
of pseudo-remainders.

\subsection{Main result}

The main cost for computing the parametric gcd via the generalized Sylvester's
theorem is charged by the computation of determinant polynomials $R_{\delta}%
$'s where $\delta\in P(d_{0},n)$. If we carry out the determinant-polynomial
calculation in a naive way, obviously the method will become very inefficient.
One may naturally ask the question: is there an analog for the
multi-polynomial case to the structure theorem for the subresultant chain of
two polynomials so that it can be used to enhance the efficiency of the
determinant-polynomial calculation required for generating all the generalized
subresultants in the generalized Sylvester's theorem? Through numerous trial
and errors, we discovered the following quadrilateral property among some
$R_{\delta}$'s where these $\delta$'s are in some sense related.

\begin{notation}
\ 

\begin{itemize}
\item For $\delta=(\delta_{1},\ldots,\delta_{n})$ and $\gamma=(\gamma
_{1},\ldots,\gamma_{n})$,
\begin{align*}
\max(\delta,\gamma)  &  =(\max(\delta_{1},\gamma_{1}),\ldots,\max(\delta
_{n},\gamma_{n}))\\
\min(\delta,\gamma)  &  =(\min(\delta_{1},\gamma_{1}),\ldots,\min(\delta
_{n},\gamma_{n}))
\end{align*}

\item $e_{i}=(0,\ldots,0,\underbrace{1}_{\text{the $i$-th slot}},0,\ldots,0)$.
\end{itemize}
\end{notation}

\begin{theorem}
[Generalized Habicht's Theorem]\label{thm:ghabicht} Let $\delta$, $\gamma\in
P(d_{0},n)$ be such that

\begin{itemize}
\item $|\delta|=|\gamma|<d_{0}$, and

\item there exist distinct $p,q$ such that $1\le p<q$ and $\delta+e_{p}%
=\gamma+e_{q}$.
\end{itemize}

\noindent Then we have
\[
\operatorname*{prem}(R_{\delta},R_{\gamma})=r_{\min\left(  \delta
,\gamma\right)  }\ R_{\max\left(  \delta,\gamma\right)  }%
\]

\end{theorem}

\begin{example}
Let $d=(3,3,4)$ and let $\delta=(1,1)$ and $\gamma=(2,0)$. Note
\begin{align*}
\delta+e_{1}  &  =\gamma+e_{2}\\
\min\left(  \delta,\gamma\right)   &  =(1,0)\\
\max\left(  \delta,\gamma\right)   &  =(2,1)
\end{align*}
Thus, by the generalized Habicht's theorem, we have
\[
\operatorname*{prem}(R_{(1,1)},R_{(2,0)})=r_{(1,0)}R_{(2,1)}%
\]

\end{example}

\begin{remark}
Let $d=\left(  3,4\right)  $. Then $P(3,1)=[3,2,1,0]$. We cannot find $\delta$
and $\gamma$ from $P(3,1)$ such that they have two different components. Thus
Theorem \ref{thm:ghabicht} is not applicable to the bi-polynomial case. We say
Theorem \ref{thm:ghabicht} is a generalization of the classical Habicht's
theorem in the sense that the theorem generalizes the idea of the classical
Habicht's theorem, that is, establishing a relationship between a subresultant
and two other consecutive subresultants.
\end{remark}

\subsection{Proof of the generalized Habicht's theorem (Theorem
\ref{thm:ghabicht})}

For verifying the correctness of the generalized Habicht's theorem (Theorem
\ref{thm:ghabicht}), we need the following two lemmas.

\begin{lemma}
\label{lem:dp_prem} Assume $A,B\in\mathcal{Z} [x]$ are of the same degree.
Then $\operatorname*{dp}(A,B)=\operatorname{prem}(A,B)$.
\end{lemma}

\begin{proof}
Immediate from a specialization of \cite[Theorem 7.5.3]{1993_Mishra} for $m=n$.
\end{proof}

\begin{lemma}
\label{lem:min_max} For all $\delta,\gamma$ satisfying $\delta+e_{p}%
=\gamma+e_{q}\in P(d_{0},n)$ where $p\ne q\ge1$, let $\eta=\max(\delta
,\gamma)$ and $\zeta=\min(\delta,\gamma)$. Let $\delta_{0}$, $\gamma_{0}$,
$\eta_{0}$ and $\zeta_{0}$ be chosen as in \eqref{eqs:delta0}. Then

\begin{enumerate}
\item $\eta_{0}=\max(\delta_{0},\gamma_{0});$

\item $\zeta_{0}\le\min(\delta_{0},\gamma_{0}).$
\end{enumerate}
\end{lemma}

\begin{proof}
We assume $\delta=(\delta_{1},\ldots,\delta_{n})$ and $\gamma=(\gamma
_{1},\ldots,\gamma_{n})$. Obviously, $\delta,\gamma$ are nonzero vectors.

\begin{enumerate}
\item Assume $\eta=(\eta_{1},\ldots,\eta_{n})$. Then
\[
\eta_{i}=\left\{
\begin{array}
[c]{ll}%
\delta_{p}+1 & \text{if}\ i=p;\\
\delta_{q} & \text{if}\ i=q;\\
\delta_{i} & \text{otherwise}.
\end{array}
\right.
\]
which implies that $\eta=\delta+e_{p}$. Hence $\eta$ is nonzero.

Assume $\lambda_{0}= \max(\delta_{0},\gamma_{0})$. We will show $\lambda
_{0}=\max_{\substack{k>0\\\eta_{k}\ne0}}(d_{k}+\eta_{k})-d_{0}$, in other
words, $\lambda_{0}+d_{0}=\max_{\substack{k>0\\\eta_{k}\ne0}}(d_{k}+\eta_{k}%
)$. It can be derived from the following:
\begin{align*}
d_{0}+\lambda_{0}  &  =d_{0}+\max(\delta_{0},\gamma_{0})\\
&  =\max(d_{0}+\delta_{0},d_{0}+\gamma_{0})\\
&  =\max(\max_{\substack{k>0\\\delta_{k}\ne0}}(d_{k}+\delta_{k}),\max
_{\substack{k>0\\\gamma_{k}\ne0}}(d_{k}+\gamma_{k}))\\
&  =\max(\max_{\substack{k>0\\k\ne p,q\\\delta_{k}\ne0}}(d_{k}+\delta
_{k}),d_{p}+\delta_{p}+1,d_{q}+\delta_{q})\\
&  =\max_{\substack{k>0\\\eta_{k}\ne0}}(d_{k}+\eta_{k})
\end{align*}
which implies that $\eta_{0}=\lambda_{0}=\max(\delta_{0},\gamma_{0})$.

\item Assume $\zeta=(\zeta_{1},\ldots,\zeta_{n})$. Then
\[
\zeta_{i}=\left\{
\begin{array}
[c]{ll}%
\delta_{p} & \text{if}\ i=p;\\
\delta_{q}-1 & \text{if}\ i=q;\\
\delta_{i} & \text{otherwise}.
\end{array}
\right.
\]

Assume $\lambda_{0}= \min(\delta_{0},\gamma_{0})$. Thus we only need to show
that $\lambda_{0}$ is no less than $\zeta_{0}$.

When $\min(\delta,\gamma)=(0,\ldots,0)$, $\zeta_{0}=1$. In this case,
$\delta=e_{q}$ and $\gamma=e_{p}$, which implies that ${\delta}_{0}%
=d_{q}+1-d_{0}$ and ${\gamma}_{0}=d_{p}+1-d_{0}$. Hence $\lambda_{0}%
=\min(d_{p},d_{q})-d_{0}+1$. Since $d_{0}=\min d$, we have
\[
\lambda_{0}=\min(d_{p},d_{q})-d_{0}+1\ge1=\zeta_{0}%
\]
When $\min(\delta,\gamma)\ne(0,\ldots,0)$, we will show $\lambda_{0}\ge
\zeta_{0}=\max_{\substack{k>0\\\zeta_{k}\ne0}}(d_{k}+\zeta_{k})-d_{0}$, or
equivalently, $\lambda_{0}+d_{0}\ge\max_{\substack{k>0\\\zeta_{k}\ne0}%
}(d_{k}+\zeta_{k})$. To achieve this goal, we consider the following three cases.

\begin{enumerate}
[(a)]

\item If $d_{0}+\delta_{0}=d_{i}+\delta_{i}>d_{p}+\delta_{p}$ where $i\ge1$
and $i\ne p,q$, then
\begin{align*}
d_{0}+\gamma_{0}  &  =\max_{\substack{k\ge1\\\gamma_{k}\ge0}}(d_{k}+\gamma
_{k})\\
&  =\max\left(  \max_{\substack{1\le k\ne p,q\\\gamma_{k}\ne0}}(d_{k}%
+\gamma_{k}), d_{p}+\gamma_{p}, d_{q}+\gamma_{q}\right) \\
&  =\max\left(  \max_{\substack{1\le k\ne p,q\\\gamma_{k}\ne0}}(d_{k}%
+\delta_{k}), d_{p}+\delta_{p}+1, d_{q}+\delta_{q}-1\right) \\
&  =\max\left(  d_{i}+\delta_{i}, d_{p}+\delta_{p}+1, d_{q}+\delta
_{q}-1\right) \\
&  =d_{i}+\delta_{i}\\
&  =d_{0}+\delta_{0}%
\end{align*}
Hence $\gamma_{0}=\delta_{0}$, which implies $\lambda_{0}=\delta_{0}$. Thus we
only need to show $\delta_{0}+d_{0}\ge\max_{\substack{k>0\\\zeta_{k}\ne
0}}(d_{k}+\zeta_{k})$, which can be verified through the following deduction:
\[
\lambda_{0}+d_{0}=\delta_{0}+d_{0}=\max_{\substack{k\ge1\\\delta_{k}\ne
0}}(d_{k}+\delta_{k})\ge\max_{\substack{k>0\\\zeta_{k}\ne0}}(d_{k}+\zeta_{k})
\]

\item If $d_{0}+\delta_{0}=d_{p}+\delta_{p}$, then
\begin{align*}
d_{0}+\gamma_{0}  &  =\max_{\substack{k\ge1\\\gamma_{k}\ge0}}(d_{k}+\gamma
_{k})\\
&  =\max\left(  \max_{\substack{1\le k\ne p,q\\\gamma_{k}\ne0}}(d_{k}%
+\gamma_{k}), d_{p}+\gamma_{p}, d_{q}+\gamma_{q}\right) \\
&  =\max\left(  \max_{\substack{1\le k\ne p,q\\\gamma_{k}\ne0}}(d_{k}%
+\delta_{k}), d_{p}+\delta_{p}+1, d_{q}+\delta_{q}-1\right) \\
&  =d_{p}+\delta_{p}+1
\end{align*}
Hence $\gamma_{0}=\delta_{0}+1$. Therefore, $\lambda_{0}=\delta_{0}$. Thus we
only need to show $\delta_{0}+d_{0}\ge\max_{\substack{k>0\\\zeta_{k}\ne
0}}(d_{k}+\zeta_{k})$, which can be verified through the following deduction:
\[
\lambda_{0}+d_{0}=\delta_{0}+d_{0}=\max_{\substack{k\ge1\\\delta_{k}\ne
0}}(d_{k}+\delta_{k})\ge\max_{\substack{k>0\\\zeta_{k}\ne0}}(d_{k}+\zeta_{k})
\]

\item If $d_{0}+\delta_{0}=d_{q}+\delta_{q}$, we consider the following two
cases.\\[10pt]Case (i): $d_{0}+\gamma_{0}=d_{q}+\gamma_{q}$. Using the same
reasoning as in (b), we can derive $\lambda_{0}+d_{0}\ge\max
_{\substack{k>0\\\zeta_{k}\ne0}}(d_{k}+\zeta_{k})$.\\[2pt]Case (ii):
$d_{0}+\gamma_{0}>d_{q}+\gamma_{q}$. Using the same reasoning as in (a), we
can derive $\lambda_{0}+d_{0}\ge\max_{\substack{k>0\\\zeta_{k}\ne0}%
}(d_{k}+\zeta_{k})$.
\end{enumerate}

To summarize, we have $\lambda_{0}=\min(\delta_{0},\gamma_{0})\ge
\max_{\substack{k>0\\\zeta_{k}\ne0}}(d_{k}+\zeta_{k})-d_{0}=\zeta_{0}$.
\end{enumerate}
\end{proof}

\begin{remark}
Here we provide an example as a witness such that the relation $>$ in Lemma
\ref{lem:min_max}-2 holds. Let $d=(3,3,4)$, $\delta=(1,1)$ and $\gamma=(2,0)$.
Then $\delta_{0}=2$ and $\gamma_{0}=2$. One may verify that $\delta
+e_{1}=\gamma+e_{2}\in P(3,2)$. With simple calculation, we have $\zeta
=\min(\delta,\gamma)=(1,0)$ and thus $\zeta_{0}=1< \min(\delta_{0},\gamma
_{0})=2$.
\end{remark}

\bigskip

\noindent Now we proceed to prove the generalized Habicht's theorem (Theorem
\ref{thm:ghabicht}).

\begin{proof}
[Proof of the generalized Habicht's theorem] Assume $\delta=(\delta_{1}%
,\ldots,\delta_{n})$ and $\gamma=(\gamma_{1},\ldots,\gamma_{n})$. Under the
constraint, we have
\[
\gamma_{p}=\delta_{p}+1, \qquad\gamma_{q}=\delta_{q}-1,\qquad\gamma_{i}%
=\delta_{i}\ \ \text{for}\ \ i\ne p,q
\]
Denote $\max\left(  \delta,\gamma\right)  $ and $\min\left(  \delta
,\gamma\right)  $ by $\eta$ and $\zeta$, respectively. Then
\begin{align*}
\eta=\max\left(  \delta,\gamma\right)   &  =(\delta_{1},\ldots,\delta
_{p-1},\underset{p}{\underbrace{\delta_{p}+1}},\delta_{p+1},\ldots,\delta
_{n})=\delta+e_{p}\\
\zeta=\min\left(  \delta,\gamma\right)   &  =(\delta_{1},\ldots,\delta
_{q-1},\underset{q}{\underbrace{\delta_{q}-1}},\delta_{q+1},\ldots,\delta
_{n}))=\delta-e_{q}%
\end{align*}
Let $\delta_{0}$, $\gamma_{0}$, $\eta_{0}$ and $\zeta_{0}$ be chosen as in
\eqref{eqs:delta0}. Hence
\[
R_{\delta}=\widetilde{R}_{(\delta_{0},\delta_{1},\ldots,\delta_{n}%
)},\ \ R_{\gamma}=\widetilde{R}_{(\gamma_{0},\gamma_{1},\ldots,\gamma_{n}%
)},\ \ R_{\eta}=\widetilde{R}_{(\eta_{0},\eta_{1},\ldots,\eta_{n}%
)},\ \ R_{\zeta}=\widetilde{R}_{(\zeta_{0},\zeta_{1},\ldots,\zeta_{n})}%
\]
respectively. By Lemma \ref{lem:min_max}, $\eta_{0}=\max(\delta_{0},\gamma
_{0})$ and $\zeta_{0}\le\min(\delta_{0},\gamma_{0})$.

\noindent For the sake of simplicity, we introduce the following notations.
Let
\[
P_{i,j}=x^{j}F_{i},\ \ \mathcal{Q}_{i,j}=P_{i,j},\ldots,P_{i,0}
\]
By Definition \ref{def:sres_npoly},
\begin{align*}
R_{\max(\delta,\gamma)}=  &  \operatorname*{dp}\,(\mathcal{Q}_{0,\max
(\delta_{0},\gamma_{0})-1}, \mathcal{Q}_{1,\delta_{1}-1},\ldots
,\ \ \ \ \ \ \ \ \ \ \ \;\mathcal{Q}_{p,\delta_{p}},\ldots
,\ \ \ \ \ \ \ \ \ \ \ \mathcal{Q}_{q,\delta_{q}-1},\ldots,\mathcal{Q}%
_{n,\delta_{n}-1})\\
=  &  \operatorname*{dp}\,(\mathcal{Q}_{0,\max(\delta_{0},\gamma_{0})-1},
\mathcal{Q}_{1,\delta_{1}-1},\ldots,\ \underbrace{P_{p,\delta_{p}}%
,\mathcal{Q}_{p,\delta_{p}-1}}_{\mathcal{Q}_{p,\delta_{p}}},\ldots
,\underbrace{P_{q,\delta_{q}-1},\mathcal{Q}_{q,\delta_{q}-2}}_{\mathcal{Q}%
_{q,\delta_{q}-1}},\ldots,\mathcal{Q}_{n,\delta_{n}-1})\\
=  &  \operatorname*{dp}\,(\mathcal{Q}_{0,\max(\delta_{0},\gamma_{0})-1},
\mathcal{Q}_{1,\delta_{1}-1},\ldots,\ P_{p,\gamma_{p}-1},\mathcal{Q}%
_{p,\delta_{p}-1},\ldots,P_{q,\delta_{q}-1},\mathcal{Q}_{q,\delta_{q}%
-2},\ldots,\mathcal{Q}_{n,\delta_{n}-1})\\
&  \qquad\text{(replacing $\delta_{p}$ with $\gamma_{p}-1$)}%
\end{align*}
Consider $c_{q(\delta_{q}-1)}c_{p(\delta_{p}-1)}^{\prime}R_{\max(\delta
,\gamma)}$. By pushing $c_{q(\delta_{q}-1)}c_{p(\gamma_{p}-1)}^{\prime}$ into
$\operatorname{dp}$, we obtain:
\begin{align*}
&  c_{q(\delta_{q}-1)}c_{p(\gamma_{p}-1)}^{\prime}R_{\max(\delta,\gamma)}\\
=  &  \operatorname*{dp}\,(\mathcal{Q}_{0,\max(\delta_{0},\gamma_{0})-1},
\mathcal{Q}_{1,\delta_{1}-1},\ldots,\ c_{p(\gamma_{p}-1)}^{\prime}%
P_{p,\gamma_{p}-1},\mathcal{Q}_{p,\delta_{p}-1},\ldots,c_{q(\delta_{q}%
-1)}P_{q,\delta_{q}-1},\mathcal{Q}_{q,\delta_{q}-2},\ldots,\mathcal{Q}%
_{n,\delta_{n}-1})
\end{align*}

\noindent Recall Proposition \ref{lem:sres_multipoly}. We have
\[
R_{\delta}=\widetilde{R}_{(\delta_{0},\delta_{1},\ldots,\delta_{n})}%
=\sum_{i=0}^{n}\sum_{j=0}^{\delta_{i}-1}c_{ij}P_{i,j},\qquad R_{\gamma
}=\widetilde{R}_{(\gamma_{0},\gamma_{1},\ldots,\gamma_{n})}=\sum_{i=0}^{n}%
\sum_{j=0}^{\gamma_{i}-1}c_{ij}^{\prime}P_{i,j}%
\]
Therefore,
\begin{align*}
c_{q(\delta_{q}-1)}P_{q,\delta_{q}-1}  &  =R_{\delta}-\sum_{(i,j)\in C_{1}%
}c_{ij}P_{i,j}\\
c_{p(\gamma_{p}-1)}^{\prime}P_{p,\gamma_{p}-1} & =R_{\gamma
}-\sum_{(i,j)\in C_{2}}c_{ij}^{\prime}P_{i,j}%
\end{align*}
where
\begin{align*}
C_{1}  &  =\{(i,j):\,0\leq i\leq n,0\leq j\leq\delta_{i}-1\}\backslash
\{(q,\delta_{q}-1)\}\\
C_{2}  &  =\{(i,j):\,0\leq i\leq n,0\leq j\leq\gamma_{i}-1\}\backslash
\{(p,\gamma_{p}-1)\}
\end{align*}

\noindent Then we have the following derivations:
\begin{align*}
&  c_{q(\delta_{q}-1)}c_{p(\gamma_{p}-1)}^{\prime}R_{\max(\delta,\gamma)}\\
=  &  \operatorname*{dp}\,(\mathcal{Q}_{0,\max(\delta_{0},\gamma_{0})-1},
\mathcal{Q}_{1,\delta_{1}-1},\ldots,\ R_{\gamma},\mathcal{Q}_{p,\delta_{p}%
-1},\ldots,R_{\delta},\mathcal{Q}_{q,\delta_{q}-2},\ldots,\mathcal{Q}%
_{n,\delta_{n}-1})\\
&  \qquad\text{(by the multi-linearity of determinant polynomials)}\\
=  &  (-1)^{\delta_{p}+\ldots+\delta_{q-1}+1}\,\operatorname*{dp}%
\,(\mathcal{Q}_{0,\max(\delta_{0},\gamma_{0})-1}, \mathcal{Q}_{1,\delta_{1}%
-1},\ldots,\ \mathcal{Q}_{p,\delta_{p}-1},\ldots,\mathcal{Q}_{q,\delta_{q}%
-2},\ldots,\mathcal{Q}_{n,\delta_{n}-1},R_{\delta},R_{\gamma})\\
&  \qquad\text{(by the alternating property of determinant polynomials)}\\
=  &  (-1)^{\delta_{p}+\ldots+\delta_{q-1}+1}\widetilde{r}_{(\max(\delta
_{0},\gamma_{0}),\delta_{1},\ldots,\delta_{q}-1,\ldots,\delta_{n}%
)}\operatorname*{dp}(R_{\delta},R_{\gamma})\\
&  \qquad\text{(using the properties of determinant polynomials of
block-triangular matrices)}%
\end{align*}
Note that
\[
\min(\delta,\gamma)=(\delta_{1},\ldots,\delta_{q}-1,\ldots,\delta_{n}),
\quad\max(\delta_{0},\gamma_{0})\ge\min(\delta_{0},\gamma_{0})\ge\zeta_{0}%
\]
and $a_{0d_{0}}=1$. It follows that
\[
\widetilde{R}_{(\max(\delta_{0},\gamma_{0}),\delta_{1},\ldots,\delta
_{q}-1,\ldots,\delta_{n})}=R_{\delta-e_{q}}=R_{\min(\delta,\gamma)}%
\]
and thus
\[
\widetilde{r}_{(\max(\delta_{0},\gamma_{0}),\delta_{1},\ldots,\delta
_{q}-1,\ldots,\delta_{n})}=r_{\min(\delta,\gamma)}%
\]
Hence
\begin{align}
c_{q(\delta_{q}-1)}c_{p(\gamma_{p}-1)}^{\prime}R_{\max(\delta,\gamma)}=  &
(-1)^{\delta_{p}+\cdots+\delta_{q-1}+1}r_{\min(\delta,\gamma)}%
\operatorname*{dp}(R_{\delta},R_{\gamma}) \label{eq:cp*cq*Smax}%
\end{align}

\noindent Recall $\zeta=\min(\delta,\gamma)$. By Lemma
\ref{lem:sres_multipoly}-3, we can further deduce that
\begin{align*}
c_{q(\delta_{q}-1)}  &  =(-1)^{\sigma_{q}}\widetilde{r}_{(\delta_{0},\zeta
_{1},\ldots,\zeta_{n})}\\
&  =(-1)^{\sigma_{q}}\widetilde{r}_{(\zeta_{0},\zeta_{1},\ldots,\zeta_{n})}
\ \ \ \text{(since $\delta_{0}\ge\zeta_{0}$ and $a_{0d_{0}}=1$)}\\
&  =(-1)^{\sigma_{q}}r_{\min(\delta,\gamma)}\\
c_{p(\gamma_{p}-1)}^{\prime}  &  =(-1)^{\sigma_{p}}\widetilde{r}_{(\gamma
_{0},\zeta_{1},\ldots,\zeta_{n})}\\
&  =(-1)^{\sigma_{p}}\widetilde{r}_{(\zeta_{0},\zeta_{1},\ldots,\zeta_{n})}
\ \ \ \text{(since $\gamma_{0}\ge\zeta_{0}$ and $a_{0d_{0}}=1$)}\\
&  =(-1)^{\sigma_{p}}r_{\min(\delta,\gamma)}%
\end{align*}
where
\begin{align*}
\sigma_{p}  &  \equiv1+\delta_{q}+\cdots+\delta_{n}\mod 2\\
\sigma_{q}  &  \equiv1+\gamma_{p}+\cdots+\gamma_{n}\mod 2
\end{align*}
Since $\delta+e_{p}=\gamma+e_{q}$,
\[
\sigma_{p}+\sigma_{q}\equiv(\delta_{p}+\cdots+\delta_{q-1}+1)\mod 2
\]
Hence
\begin{align}
c_{q(\delta_{q}-1)}c_{p(\gamma_{p}-1)}^{\prime}  &  =(-1)^{\delta_{p}%
+\cdots+\delta_{q-1}+1}r_{\min(\delta,\gamma)}^{2} \label{eq:cp*cq}%
\end{align}
The substitution of \eqref{eq:cp*cq} into \eqref{eq:cp*cq*Smax} leads to
\begin{align*}
&  (-1)^{\delta_{p}+\cdots+\delta_{q-1}+1}r_{\min(\delta,\gamma)}^{2}%
R_{\max(\delta,\gamma)}= (-1)^{\delta_{p}+\cdots+\delta_{q-1}+1}r_{\min
(\delta,\gamma)}\operatorname*{dp}(R_{\delta},R_{\gamma})
\end{align*}
It is easy to know that $r_{\min(\delta,\gamma)}\neq0$.\footnote{Otherwise, by
the generalized Sylvester's theorem, the following never occurs:
\begin{align*}
(\deg G_{1},\ldots,\deg G_{n})  &  =(d_{0}-\zeta_{1},d_{0}-(\zeta_{1}%
+\zeta_{2}),\ldots,d_{0}-(\zeta_{1}+\cdots+\zeta_{n}))
\end{align*}
which is obviously not true.} The cancellation of the common factors from both
sides immediately yields
\begin{equation}
r_{\min(\delta,\gamma)}R_{\max(\delta,\gamma)}=\operatorname*{dp}(R_{\delta
},R_{\gamma}) \label{eq:sminSmax}%
\end{equation}

\noindent Observing that $r_{\delta}$ and $r_{\gamma}$ are not identical to
zero\footnote{Otherwise, by the generalized Sylvester's theorem, the following
two never occur:
\begin{align*}
(\deg G_{1},\ldots,\deg G_{n})  &  =(d_{0}-\delta_{1},d_{0}-(\delta_{1}%
+\delta_{2}),\ldots,d_{0}-(\delta_{1}+\cdots+\delta_{n}))\\
\text{or}\quad(\deg G_{1},\ldots,\deg G_{n})  &  =(d_{0}-\gamma_{1}%
,d_{0}-(\gamma_{1}+\gamma_{2}),\ldots,d_{0}-(\gamma_{1}+\cdots+\gamma_{n}))
\end{align*}
which is obviously not true.} and combining the assumption that $\delta
+e_{p}=\gamma+e_{q}$, we have
\[
\deg R_{\delta}=d_{0}-(\delta_{1}+\cdots+\delta_{n})=d_{0}-(\gamma_{1}%
+\cdots+\gamma_{n})=\deg R_{\gamma}%
\]
Applying Lemma \ref{lem:dp_prem} to \eqref{eq:sminSmax}, we obtain
\[
\operatorname*{prem}(R_{\delta},R_{\gamma})=r_{\min\left(  \delta
,\gamma\right)  }\ R_{\max\left(  \delta,\gamma\right)  }%
\]

\end{proof}

\subsection{An algorithm for computing parametric gcd with the generalized
Habicht's theorem}

In this subsection, we will develop an efficient algorithm for computing
parametric gcd by exploiting the generalized Habicht's theorem
(Theorem~\ref{thm:ghabicht}). Let us begin by observing the following.

\begin{lemma}
\label{lem:consecutive_degree} Let $\delta$, $\gamma$ satisfy the conditions
in the generalized Habicht's theorem (Theorem~\ref{thm:ghabicht}). Then we have

\begin{enumerate}
\item $\deg R_{\delta}=\deg R_{\gamma}\ge1,$

\item $\deg R_{\min\left(  \delta,\gamma\right)  }=\deg R_{\delta}+1$, and

\item $\deg R_{\max\left(  \delta,\gamma\right)  }=\deg R_{\delta}-1$.
\end{enumerate}
\end{lemma}

\begin{proof}
\ 

\begin{enumerate}
\item $\deg R_{\delta}=\deg R_{\gamma}\ge1$. The claim can be obtained from
\begin{align*}
\deg R_{\delta}=d_{0}-|\delta|=d_{0}-|\gamma|=\deg R_{\gamma}\ge1
\end{align*}

\item $\deg R_{\min\left(  \delta,\gamma\right)  }=\deg R_{\delta}+1$

Assume $\delta=(\delta_{1},\ldots,\delta_{n})$, $\gamma=(\gamma_{1}%
,\ldots,\gamma_{n})$ and $p,q$ are such that $1\le p<q$ and $\delta
+e_{p}=\gamma+e_{q}$. Then we have
\begin{align*}
\deg R_{\min\left(  \delta,\gamma\right)  }  &  =d_{0}-(\delta_{1}%
+\cdots+\delta_{q-1}+(\delta_{q}-1)+\delta_{q+1}+\cdots+\delta_{n})\\
&  =d_{0}-(\delta_{1}+\cdots+\delta_{n})+1\\
&  =\deg R_{\delta}+1
\end{align*}

\item $\deg R_{\max\left(  \delta,\gamma\right)  } =\deg R_{\delta}-1$

We have%
\begin{align*}
\deg R_{\max\left(  \delta,\gamma\right)  }  &  =d_{0}-(\delta_{1}%
+\cdots+\delta_{p-1}+(\delta_{p}+1)+\delta_{p+1}+\cdots+\delta_{n})\\
&  =d_{0}-(\delta_{1}+\cdots+\delta_{n})-1\\
&  =\deg R_{\delta}-1
\end{align*}
\end{enumerate}
\vspace{-2em}
\end{proof}


From Lemma \ref{lem:consecutive_degree}, we observe that the smaller
$|\delta|$ is, the bigger $\deg R_{\delta}$ is. Meanwhile, from Definition
\ref{def:sres_npoly}, the matrix for defining $R_{\delta}$ is with smaller
order, which makes $R_{\delta}$ easier to be computed. Therefore, we may
follow the strategy for computing subresultants of two univariate polynomials
with the classical Habicht theorem and devise an analogy to compute
subresultants of several univariate polynomials with the generalized Habicht's
theorem. More explicitly, one can compute $R_{\delta},R_{\gamma}%
,R_{\min\left(  \delta,\gamma\right)  }$ and $R_{\max\left(  \delta
,\gamma\right)  }$ with the order illustrated below.
\[%
\begin{matrix}
R_{\delta} & \quad & R_{\max\left(  \delta,\gamma\right)  }\\
&  & \\
\ R_{\min\left(  \delta,\gamma\right)  } & \quad & R_{\gamma}%
\end{matrix}
\]
More explicitly, the subresultant in the upper-right corner (with lower
degree) can be computed from the three at the other corners (with higher
degrees) by using the generalized Habicht's theorem. For example, in order to
compute all the $R_{\delta}$'s of three polynomials where $\delta\in
P(d_{0},n)$, we follow the orders indicated by the following array:
\[%
\begin{matrix}
\vdots & \vdots & \vdots & \\
R_{(0,2)} & R_{(1,2)} & R_{(2,2)} & \cdots\\[10pt]%
R_{(0,1)} & R_{(1,1)} & R_{(2,1)} & \cdots\\[10pt]%
R_{(0,0)} & R_{(1,0)} & R_{(2,0)} & \cdots\\
&  &  &
\end{matrix}
\]
Roughly speaking, for every rectangle, the upper-right subresultant can be
computed by taking the pseudo-remainder of the upper-left and the bottom-right
and then divided by the leading coefficient of the bottom-left.

\bigskip

\noindent Now  we  package together all the above  findings into   an algorithm.
 See
Algorithm~\ref{alg:non_recursive_cute} EPGCD. The name stands for algorithm for
 \textbf{E}fficient \textbf{P}arametric \textbf{GCD} .

\begin{figure}[pt]

\hrule

\bigskip

\begin{algorithm}
\label{alg:non_recursive_cute} $\mathcal{G} \leftarrow$ EPGCD($d$)

\begin{itemize}
\item[In \ :] $d$ is a list of degrees

\item[Out:] $\mathcal{G}$ is a representation of the parametric gcd for $d$,
that is, $\mathcal{G}=\left(  \left(  H_{1},G_{1}\right)  ,\ldots,\left(
H_{p},G_{p}\right)  \right)  $ such that%
\[
\gcd\left(  F\right)  =\left\{
\begin{array}
[c]{lll}%
G_{1} & \text{if} & H_{1}\neq0\\
G_{2} & \text{else if} & H_{2}\neq0\\
\ \vdots & \ \ \vdots & \ \ \ \vdots\\
G_{p} & \text{else if} & H_{p}\neq0
\end{array}
\right.
\]

\end{itemize}

\begin{enumerate}
\item $F\leftarrow(F_{0},\ldots,F_{n})$ where $n=\#d-1$ and $F_{i}=\sum
_{j=0}^{d_{i}}a_{ij}x^{j}\ $(where $a_{ij}$ are indeterminates with
$a_{0d_{0}}=1$)

\item $R_{(0,\ldots,0)}\leftarrow F_{0}$,\;\; $r_{(0,\ldots,0)}\leftarrow
\mathrm{lc}_{x}(R_{(0,\ldots,0)})$

\item $\mathcal{G}\leftarrow\left\{  \left(  r_{(0,\ldots,0)},R_{(0,\ldots
,0)}\right)  \right\}  $

\item For $i=1,\ldots,d_{0}$ do

\begin{enumerate}
[a.]

\item $P_{i}\leftarrow[\delta: |\delta|=i]$ ordered decreasingly under
$\succ_{\mathrm{glex}}$

\item For $\delta=(\delta_{1},\ldots,\delta_{n})\in P_{i}$ do

\begin{enumerate}
[(1)]

\item $\mathcal{I}\leftarrow\{{k}:\,\delta_{k}\ne0\}$

\item If $\#\mathcal{I}=1$, that is, $\mathcal{I}=\{k\}$, then

\begin{enumerate}
[(a)]

\item if $\delta_{k}=1$, then

\quad$R_{\delta}\leftarrow\operatorname{prem}(F_{k},F_{0})$

\item if $\delta_{k}=2$, then

\quad$\gamma\leftarrow\delta-2e_{k}$,\ \ \ $\zeta\leftarrow\delta-e_{k}$
\newline

\quad$\delta_{0}\leftarrow\left\{
\begin{array}
[c]{ll}%
\max\limits_{\substack{i\ge1\\\delta_{i}\ne0}}(d_{i}+\delta_{i})-d_{0} &
\text{if}\ \ \exists_{i\ge1}\delta_{i}\ne0;\\
1 & \text{otherwise}.
\end{array}
\right.  $

\quad$R_{\delta}\leftarrow\operatorname{prem}(R_{\gamma},R_{\zeta
})\bigg/\left(  r_{\gamma}\right)  ^{\delta_{0}-1}$

\item if $\delta_{k}>2$, then

\quad$\gamma\leftarrow\delta-2e_{k}$,\ \ \ $\zeta^{\prime}\leftarrow
\delta-e_{k}$

\quad$R_{\delta}\leftarrow\operatorname{prem}(R_{\gamma},R_{\zeta
})\bigg/\left(  r_{\gamma}\right)  ^{2}$
\end{enumerate}

\item Else

\begin{enumerate}
[(a)]

\item Choose $p,q\in\mathcal{I}$ such that $p<q$

\item $\gamma\leftarrow\delta-e_{p}$,\ \ \ $\eta\leftarrow\delta-e_{q}$,
\ \ \ $\zeta\leftarrow\delta-e_{p}-e_{q}$

\item $R_{\delta}\leftarrow\operatorname*{prem}\left(  R_{\gamma},R_{\eta
}\right)  /r_{\zeta}$
\end{enumerate}

\item $r_{\delta}\leftarrow\mathrm{lc}_{x}(R_{\delta})$

\item $\mathcal{G} \leftarrow$ append $\left\{  \left(  r_{\delta},R_{\delta
}\right)  \right\}  $ to $\mathcal{G}$
\end{enumerate}
\end{enumerate}

\item Return $\mathcal{G}$.
\end{enumerate}
\end{algorithm}
\medskip
\hrule

\end{figure}
\bigskip

\noindent The correctness of the  algorithm relies on the following lemmas.

\begin{lemma}
\label{lem:mod_habicht} Given $\delta=(0,\ldots,0,\delta_{i},0,\ldots,0)$
where $\delta_{i}>0$,

\begin{enumerate}
\item \label{lem:mod_habicht_1}if $\delta_{i}=1$, then $R_{\delta
}=\operatorname{prem}(F_{i},F_{1})$;

\item \label{lem:mod_habicht_2}if $\delta_{i}=2$, then
\[
\left(  r_{\gamma}\right)  ^{\delta_{0}-1}\cdot R_{\delta}=\operatorname{prem}%
\left(  R_{\gamma},R_{\zeta}\right)
\]
where $\delta_{0}$ is as in \eqref{eqs:delta0}, $\gamma=\delta-2e_{i}$ and
$\zeta=\delta-e_{i}$.

\item \label{lem:mod_habicht_3}if $\delta_{i}>2$, then
\[
\left(  r_{\gamma}\right)  ^{2}\cdot R_{\delta}=\operatorname{prem}\left(
R_{\gamma},R_{\zeta}\right)
\]
where $\gamma=\delta-2e_{i}$ and $\zeta=\delta-e_{i}$.
\end{enumerate}
\end{lemma}

\begin{proposition}
Algorithm \ref{alg:non_recursive_cute} terminates in finite steps and the
output is correct.
\end{proposition}

\begin{proof}
The termination of Algorithm \ref{alg:non_recursive_cute} is guaranteed by the
finiteness of the partition set $P(d_{0},n)$.

To show the correctness of the algorithm, we should first make it clear that a
specific order as indicated in Step 4 should be used for computing $R_{\delta
}$'s where $\delta\in P(d_{0},n)$. This is the key for the algorithm to proceed.

When $|\delta|=0$, the only possibility for $\delta$ is $\delta=(0,\ldots,0)$.
Thus the corresponding $R_{\delta}$ is $F_{1}$, which is obvious (see Step 2).
Next we should compute $R_{\delta}$'s for $\delta\ne(0,\ldots,0)$. For this
purpose, we divide $\delta$'s into several groups $P_{1}, \ldots,P_{n}$ where
$P_{i}=\{\delta:|\delta|=i\}$. We will compute $R_{\delta}$ for $\delta$'s in
$P_{1}$ first, and then $P_{2}$,...., and so on because the computation of
$R_{\delta}$ for $\delta$'s in $P_{i}$ relies on $R_{\delta}$ for $\delta$'s
in $P_{i-1}$ (and $P_{i-2}$ when $i\ge2$).

In order to compute $R_{\delta}$ for $\delta$'s in $P_{i}$, we need to
identify whether $R_{\delta}$ involves only two polynomials.

\begin{enumerate}
\item In the affirmative case, the set $\mathcal{I}=\{{k}:\,\delta_{k}\ne0\}$
contains only one element and thus $\delta$'s are of the form $(0,\ldots
,0,\underbrace{\delta_{k}}_{\text{the $k$-th slot}},0,\ldots,0)$. The
computation technique of $R_{\delta}$ is determined by $\delta_{k}$.

\begin{itemize}
\item If $\delta_{k}=1$, by Lemma \ref{lem:mod_habicht}%
-\ref{lem:mod_habicht_1}, $R_{\delta}$'s are exactly $\mathrm{prem}%
(F_{k},F_{1})$. See Step 4b(2)(a).

\item If $\delta_{k}=2$, by Lemma \ref{lem:mod_habicht}%
-\ref{lem:mod_habicht_2},
\[
R_{\delta}=\operatorname{prem}\left(  R_{\gamma},R_{\zeta}\right)
\big/\left(  r_{\gamma}\right)  ^{\delta_{0}-1}%
\]
where $\delta_{0}$ is as in \eqref{eqs:delta0}, $\gamma=\delta-2e_{k}$ and
$\zeta=\delta-e_{k}$. See Step 4b(2)(b).

\item If $\delta_{k}>2$, by Lemma \ref{lem:mod_habicht}%
-\ref{lem:mod_habicht_3},
\[
R_{\delta}=\operatorname{prem}\left(  R_{\gamma},R_{\zeta}\right)
\big/\left(  r_{\gamma}\right)  ^{2}%
\]
where $\gamma=\delta-2e_{k}$ and $\zeta=\delta-e_{k}$. See Step 4b(2)(c).
\end{itemize}

\item In the negative case, we use the generalized Habicht's theorem to
compute $R_{\delta}$. By Theorem \ref{thm:ghabicht}, we need to find
$\delta-e_{p}$ and $\delta-e_{q}$ such that $1\le p< q$ first. Since
$\mathcal{I}=\{{k}:\,\delta_{k}\ne0\}$ contains at least two elements, we set
$p$ and $q$ to be two distinct elements chosen from $\mathcal{I}$ such that
$p<q$. For the sake of simplicity, we use the short-hand notations:
$\gamma=\delta-e_{p}$, $\eta=\delta-e_{q}$. By Theorem \ref{thm:ghabicht},
\[
R_{\max\left(  \gamma,\eta\right)  }=\operatorname*{prem}(R_{\gamma},R_{\eta
})/r_{\min\left(  \gamma,\eta\right)  }\
\]

The above computations are carried out in Step 4b(3).
\end{enumerate}

The correctness of the algorithm is completed.
\end{proof}

\section{Performance}

\label{sec:performance} In this section, we compare the performance of the
proposed algorithms and the previous algorithm. \textcolor{red}{}

\subsection{Review on the recursive approach for computing parametric gcds}

For the convenience of the readers, we review a folk-lore recursive algorithm
for computing the parametric gcds for several polynomials, which is obtained
through imitating the overall schemes developed by~\cite{1996_Yang_Hou_Zeng}
for complex/real root classification.

\begin{algorithm}
\label{alg:recursive} $\mathcal{G} \leftarrow$ PGCD\_recursive($d$)

\begin{itemize}
\item[In \ :] $d$ is a list of degrees

\item[Out:] $\mathcal{G}$ is a representation of the parametric gcd for $d$,
that is,%
\[
\mathcal{G}=\left(  \left(  C_{1},G_{1}\right)  ,\ldots,\left(  C_{p}%
,G_{p}\right)  \right)
\]
such that%
\[
\gcd\left(  F\right)  =\left\{
\begin{array}
[c]{lll}%
G_{1} & \text{if} & C_{1}\\
G_{2} & \text{if} & C_{2}\\
\ \vdots & \ \vdots & \ \vdots\\
G_{p} & \text{if} & C_{p}%
\end{array}
\right.
\]

\end{itemize}

\begin{enumerate}
\item $F\leftarrow(F_{0},\ldots,F_{n})$ where $n=\#d-1$ and $F_{i}=\sum
_{j=0}^{d_{i}}a_{ij}x^{j}\ \ \ $(where $a_{ij}$ are indeterminates)

\item $\mathcal{G}\leftarrow$ PGCD\_poly\_recursive($F$). Return $\mathcal{G}
$.
\end{enumerate}
\end{algorithm}

\begin{algorithm}
$\mathcal{G}\leftarrow$ PGCD\_poly\_recursive($F$)

\begin{itemize}
\item[In \ :] $F$ is a list of polynomials

\item[Out:] $\mathcal{G}$ is a parametric gcd for $F$
\end{itemize}

\begin{enumerate}
\item $R_{i}\longleftarrow R_{i}\left(  F_{0},F_{1}\right)  \ $for
$i=0,\ldots,d_{0}$

$r_{i}\longleftarrow\mathrm{coeff}\left(  R_{i},x^{i}\right)  $

\item If $\#F=2$ then return $\left\{  \left(  r_{d_{0}}=0\wedge\cdots\wedge
r_{d_{0}-i+1}=0\wedge r_{d_{0}-i}\neq0,R_{d_{0}-i}\right)  :\,i=0,\ldots,d_{0}
\right\}  $


\item For $i\ $from $0$ to $d_{0}$

\qquad$\mathcal{G}^{\prime}\longleftarrow$PGCD\_poly\_recursive($\left[
R_{i},F_{2},\ldots,F_{n}\right]  $)

\qquad$\mathcal{G}\longleftarrow\left\{  \left(  r_{d_{0}}=0\wedge\cdots\wedge
r_{d_{0}-i+1}=0\wedge r_{d_{0}-i}\neq0\wedge C^{\prime},G^{\prime}\right)
:\left(  C^{\prime},G^{\prime}\right)  \in\mathcal{G}^{\prime}\right\}  $

\item Return $\mathcal{G}$.
\end{enumerate}
\end{algorithm}

\subsection{Comparison on degrees of polynomials in the conditions}

It is easy to determine that the number of polynomials used in the condition
generated by the $PGCD\_recursive$ algorithm is $\# P(d_{0},n)$ and so is that
generated by the $EPGCD$ algorithm. Hence the complexity of the conditions is
determined by the degree of polynomials involved in the conditions.

We use the maximal degree of all the possible $\gcd(F)$ in terms of $a$ to
capture the complexity of parametric gcds and polynomials involved in the
conditions (whose degrees are equal to the corresponding gcds). For the
comparison purpose, we introduce two notations $d_{\text{proposed}}$ and
$d_{\text{{recursive}}}$ to denote the maximal degrees of parametric gcd
polynomials produced respectively by the proposed $EPGCD$ algorithm (Algorithm
\ref{alg:non_recursive_cute}) and the folk-lore $PGCD\_recursive$ algorithm
(Algorithm \ref{alg:recursive}).

The 3-5 th columns of Table \ref{tab:comparison} show the degree comparison.
When carrying out the experiments, we terminate the procedure if it cannot be
completed in 600 seconds. In this case, we cannot get information on the
maximum degree or the ratio of degrees. Thus, we use ``?" as an indication.
From Table \ref{tab:comparison}, we observe that

\begin{enumerate}

\item Overall, the degrees of the input polynomials have a higher impact on
$d_{\text{recursive}}$ and $d_{\text{proposed}}$ than the number of input polynomials.

\item When $n$ is fixed and $d$ increases, the ratio increases dramatically,
which reveals that Algorithm \ref{alg:non_recursive_cute} is less sensitive to
$d_{i}$'s than Algorithm \ref{alg:recursive}. Thus we can claim that Algorithm
\ref{alg:non_recursive_cute} has a better scalability compared to Algorithm
\ref{alg:recursive}.

\end{enumerate}

\begin{table}[tbh]
\caption{Comparison on the performance of Algorithms
\ref{alg:non_recursive_cute} and \ref{alg:recursive}. }%
\label{tab:comparison}
\begin{center}%
\begin{tabular}
[c]{|c|c||c|c|r||r|r|}\hline
&  &  &  &  &  & \\[-8pt]%
{No.} & \textbf{$d$} & $d_{\text{recursive}}$ & $d_{\text{proposed}}$ &
$\dfrac{d_{\text{recursive}}}{d_{\text{proposed}}}$ & $t_{\text{recursive}}$ &
$t_{\text{proposed}}$\\[8pt]\hline
1 & (3, 4, 4) & 21 & 7 & 3.0 & 0.215 & 0.081\\\hline
2 & (3, 4, 5) & 26 & 8 & 3.3 & 1.143 & 0.106\\\hline
3 & (4, 4, 4) & 25 & 8 & 3.1 & 5.809 & 0.125\\\hline
4 & (4, 4, 5) & 31 & 9 & 3.4 & 120.812\  & 0.196\\\hline
5 & (5, 5, 5) & ? & 10 & ? & $>$ 600.000 & 3.343\\\hline
6 & (3, 3, 3, 3) & 16 & 6 & 2.7 & 2.743 & 0.046\\\hline
7 & (3, 3, 3, 4) & 21 & 7 & 3.0 & 37.553 & 0.056\\\hline
8 & (4, 4, 5, 5) & ? & 9 & ? & $>$ 600.000 & 0.706\\\hline
9 & (3, 3, 3, 3, 3) & 16 & 6 & 2.7 & $>$ 600.000 & 0.072\\\hline
10 & (3, 3, 4, 4, 4) & ? & 7 & ? & $>$ 600.000 & 0.131\\\hline
11 & (3, 3, 3, 4, 4, 4) & ? & 7 & ? & $>$ 600.000 & 0.156\\\hline
12 & (4, 4, 4, 5, 5, 5) & ? & 9 & ? & $>$ 600.000 & 4.181\\\hline
\end{tabular}
\end{center}
\end{table}

\subsection{Comparison on computing time}

In order to compare the computing time, we carry out a series of experiments
for a collection of parametric polynomials. The experiments are performed on a
laptop equipped with a CPU of Intel(R) Core(TM) i7-7500U and a RAM of 8GB. The
experimental results are shown in the last two columns of Table
\ref{tab:comparison} where $t_{\text{proposed}}$ and $t_{\text{recursive}}$
are the computing times of the proposed algorithm~$EPGCD$%
~\ref{alg:non_recursive_cute} and the folk-lore algorithm~$PGCD\_recursive$
\ref{alg:recursive}, respectively. From the table, we make the following observations.

\begin{enumerate}
\item $t_{\text{proposed}}<< t_{\text{recursive}}$ for most of the tested
examples, especially when the size of the problem is big.

\item The computing time of the folk-lore algorithm $PGCD\_recursive$ is
extremely sensitive to $d$, especially to $d_{0}$, while the proposed
algorithm $EPGCD$ is less sensitive.

\item Note that the size of Example 5 is apparently smaller than Example 12
but the computing time of the former is higher than that of the latter. We
attributes the inconsistency to an internal mechanism in the algorithm, where
the value of $d_{0}$ has a greater impact on the performance than other
parameters. In Example 5, $d_{0}=5$ while in Example 12, $d_{0}=4$.
\end{enumerate}

\subsection{Comparison with generalized-subresultant-matrix-based methods}

In this subsection, we make a comparison between the
generalized-subresultant-matrix methods and the method proposed in the current
paper from two aspects, including

\begin{itemize}
\item the number of determinants for partitioning the parameter space, and

\item the availability of explicit expressions for the parametric gcds.
\end{itemize}

For describing the complexity, we will use the notation $m=\max d$.

\begin{itemize}
\item In \cite{1978_Vardulakis_Stoyle}, Vardulakis and Stoyle proposed a
generalization of Sylvester matrix whose order is $(n+1)m\times2m$ and proved
that the rank deficiency of the formulated matrix is the degree of the gcd.
However, they \emph{did not} provide explicit expressions for the parametric gcds.

\item In \cite{1976_Kakie}, Kakie proposed a method for formulating the
Sylvester subresultant matrix for several polynomials and gave a condition
that these polynomials have a gcd of degree $k$ by using the rank deficiency
of the constructed matrix. In \cite{Ho:89}, Ho proved that the
\emph{parametric} gcd can be written as a determinant polynomial of the
maximal independent rows of the $0$-subresultant matrix of the input
polynomials. Ho's method produces almost the same parametric gcds as produced
by the proposed method in the current paper. The only difference is that the
parametric gcds given by Ho's method sometimes have redundant constant factors
and thus they may have higher degrees in coefficient parameters. Concerning
the number of determinants for partitioning the parameter space, the
conditions therein are much more than what are needed in the current paper.
Roughly speaking, the conditions in the current paper form a small subset of
the condition set produced by Ho's method.

\item In \cite{1971_BARNETT}, Barnett integrated the companion resultant
matrix of several polynomials together and gave a condition to determine the
degree of the gcd for the input polynomials. The condition is described with
the rank of the obtained matrix whose order is $m\times m(n+1)$. He made a
keen observation that if the rank of the matrix is $r$, then the last $r$ rows
must be independent. Based on this, he presented a method for computing the
gcd of several \emph{numeric} polynomials.

\item In \cite{2002_Diaz_Toca_Gonzalez_Vega}, Diaz-Toca and Gonzalez-Vega
generalized Barnett's result to B\'ezout matrix, hybrid B\'ezout matrix and
Hankel matrix and showed that these matrix have the same property as the
companion resultant matrix of several univariate polynomials. Based on this,
they presented a series of methods for computing the gcd of several
\emph{numeric} polynomials.

\item In the current paper, from the generalized Sylvester's theorem, it is
easy to see that the number of determinants to partition the parameter space
is exactly the number of possibilities for partition $d_{0}, d_{0}-1,\ldots0$
into $n$ parts with the appearance of $0$'s allowed in the partition.
Therefore, the number of partitions is ${\binom{d_{0}+n}{n}}$. The explicit
expressions for parametric gcds are provided in the theorem.
\end{itemize}

In Table \ref{tab:num_dets}, we compare the numbers of determinants required
for partitioning the parameter space. The table lists formulas for these
numbers and their values for several instances of $(d_{0},m,n)$. The bottom
Figure \ref{fig:comparison} shows, in the logarithmic scale, the ratios of the
numbers of determinants required by the previous methods and the proposed
method when $d_{0}$ and $m$ are fixed to be $8$ and $15$, respectively. It is
seen that the number of determinants required by the proposed method is
significantly smaller than those by the previous methods.

\begin{table}[tbh]
\caption{Comparison of $\#$ determinants required for partitioning the
parameter space}%
\label{tab:num_dets}
\centering
{\small \medskip%
\begin{tabular}
[c]{|c|c|c|c|c|}\hline
& Sylvester & Barnett/(Hybrid) B\'{e}zout/Hankel & Sylvester & Proposed\\
& Vardulakis \& Stoyle & Barnett, Diaz-Toca \& Gonzalez-Vega & Kakie \& Ho &
\\
$(d_{1},m,n)$ & $\sum\limits_{k=2m-d_{1}}^{2m}{\binom{m(n+1)}{k}}{\binom
{2m}{k}}$ & $\sum\limits_{k=m-d_{1}}^{m}{\binom{m(n+1)}{k}}$ & $\sum
\limits_{k=0}^{d_{1}}{\binom{d_{0}-k+n}{n}}{\binom{m+d_{0}}{k}}$ &
${\binom{d_{0}+n}{n}}$\\\hline
(3,4,2) & \multicolumn{1}{|r|}{77055} & \multicolumn{1}{|r|}{793} &
\multicolumn{1}{|r|}{150} & \multicolumn{1}{|r|}{10}\\\hline
(3,5,2) & \multicolumn{1}{|r|}{1114828} & \multicolumn{1}{|r|}{4928} &
\multicolumn{1}{|r|}{198} & \multicolumn{1}{|r|}{10}\\\hline
(3,6,2) & \multicolumn{1}{|r|}{13984880} & \multicolumn{1}{|r|}{31008} &
\multicolumn{1}{|r|}{256} & \multicolumn{1}{|r|}{10}\\\hline
(3,4,3) & \multicolumn{1}{|r|}{573222} & \multicolumn{1}{|r|}{2516} &
\multicolumn{1}{|r|}{209} & \multicolumn{1}{|r|}{20}\\\hline
(3,5,3) & \multicolumn{1}{|r|}{16835406} & \multicolumn{1}{|r|}{21679} &
\multicolumn{1}{|r|}{268} & \multicolumn{1}{|r|}{20}\\\hline
(3,6,3) & \multicolumn{1}{|r|}{449751660} & \multicolumn{1}{|r|}{189750} &
\multicolumn{1}{|r|}{338} & \multicolumn{1}{|r|}{20}\\\hline
(3,4,4) & \multicolumn{1}{|r|}{2699634} & \multicolumn{1}{|r|}{6195} &
\multicolumn{1}{|r|}{280} & \multicolumn{1}{|r|}{35}\\\hline
(3,5,4) & \multicolumn{1}{|r|}{130053385} & \multicolumn{1}{|r|}{68380} &
\multicolumn{1}{|r|}{351} & \multicolumn{1}{|r|}{35}\\\hline
(3,6,4) & \multicolumn{1}{|r|}{5872564815} & \multicolumn{1}{|r|}{767746} &
\multicolumn{1}{|r|}{434} & \multicolumn{1}{|r|}{35}\\\hline
(3,6,5) & \multicolumn{1}{|r|}{45949195348} & \multicolumn{1}{|r|}{2390829} &
\multicolumn{1}{|r|}{545} & \multicolumn{1}{|r|}{56}\\\hline
\end{tabular}
}\end{table}

\begin{figure}[t]
\centering
\includegraphics[width=0.65\textwidth]{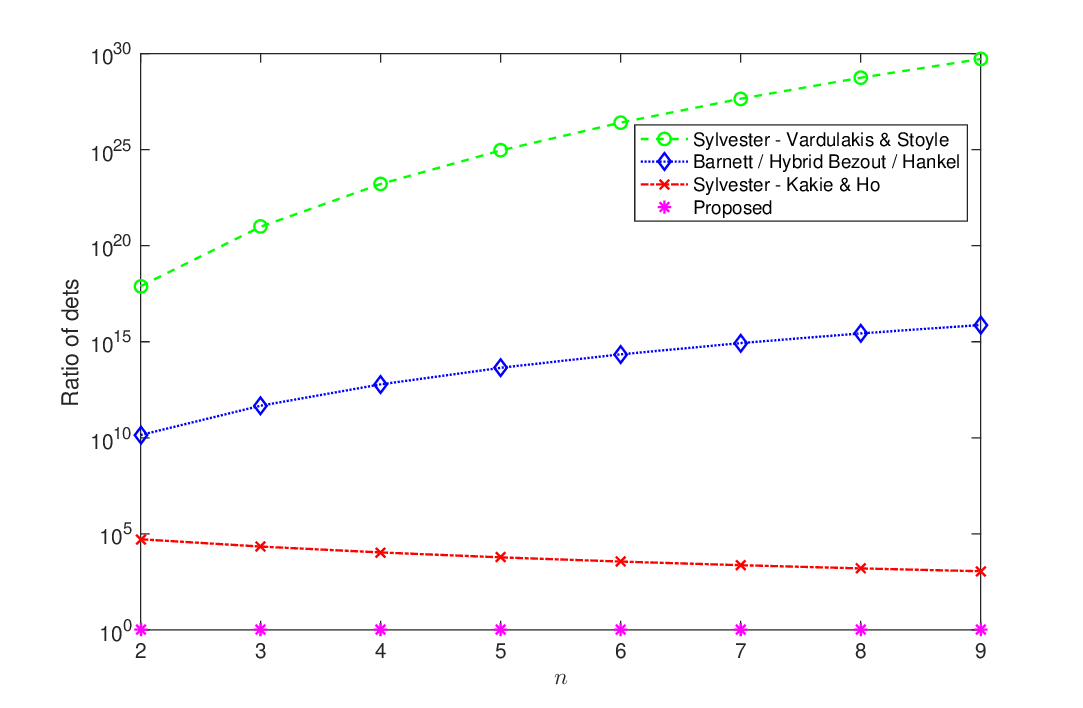}\caption{Ratios of
determinants required for partitioning the parameter space.}%
\label{fig:comparison}%
\end{figure}

\bigskip\noindent\textbf{Acknowledgements.} Hoon Hong's work was supported by
National Science Foundations of USA (Grant Nos: CCF 2212461 and CCF 2331401).
Jing Yang's work was supported by National Natural Science Foundation of China
(Grant Nos.: 12261010 and 12326353) and the Natural Science Cultivation
Project of GXMZU (Grant No.: 2022MDKJ001).


\def\cprime{$'$}

\end{document}